\documentclass[11 pt]{article}
\usepackage{amssymb,latexsym,amsmath}
\usepackage{amsthm, authblk, cite}
\usepackage{graphicx}
\usepackage{relsize}
\usepackage{enumitem}
\usepackage{appendix}

\theoremstyle{plain}
\newtheorem{theorem}{Theorem}

\newtheorem{proposition}[theorem]{Proposition}
\newtheorem{lemma}[theorem]{Lemma}

\newtheorem{procedure}{Procedure}

\theoremstyle{definition}

\newtheorem{example}[theorem]{Example}
\newtheorem{notation}[theorem]{Notation}
\newtheorem{remark}[theorem]{Remark}

\newtheorem{algorithm}{Algorithm}
\newtheorem{subalgorithm}{Subalgorithm}

\newcommand{\numberset}{\mathbb}

\newcommand{\F}{\numberset{F}}

\usepackage{hyperref}

\usepackage[margin=3cm]{geometry}

\newcommand{\FF}{{\mathbb F}}

\DeclareMathOperator{\divv}{div}

\DeclareMathOperator{\monic}{monic}
\DeclareMathOperator{\Subalg}{Subalg}
\DeclareMathOperator{\Poly}{Poly}

\title{Scalar multiplication in compressed coordinates in the trace-zero subgroup}

\author{Giulia Bianco and Elisa Gorla\thanks{The research reported in this paper was partially supported by the Swiss National Science Foundation under grant no. 200021\_150207.}}

\affil{Institut de Math\'{e}matiques, Universit\'{e} de Neuch\^{a}tel\\Rue Emile-Argand 11, CH-2000 Neuch\^{a}tel, Switzerland}

\date{}

\begin{document}
\maketitle

\begin{abstract}
\noindent
We consider trace-zero subgroups of elliptic curves over a degree three field extension. 
The elements of these groups can be represented in compressed coordinates, i.e. via the two coefficients of the line that passes through the point and its two Frobenius conjugates. 
In this paper we give the first algorithm to compute scalar multiplication in the degree three trace-zero subgroup using these coordinates. 
\end{abstract}

\section*{Introduction}\label{intr0}

Given an elliptic curve $E$ defined over a finite field $\FF_q$, an odd prime $n$ and the group $E(\FF_{q^n})$ of $\FF_{q^n}$-rational points of $E$, the trace-zero subgroup $T_n$ of $E(\FF_{q^n})$ consists of the $\FF_{q^n}$-rational points of $E$ whose trace is zero. 
Trace-zero subgroups were first proposed for cryptographic applications by Frey in \cite{Frey},
and they turn out to provide good security, efficient computation, and optimal data storage.

It is easy to show that solving the DLP in $T_n$ is as hard as solving the DLP in the entire group $E(\FF_{q^n})$ (see e.g. \cite[Proposition 1]{EM1}). Moreover, if $E$ is supersingular, an analogous result holds for the security parameter in the contest of pairing-based cryptography (see \cite{rs02} and \cite{rs09}). In particular, the cardinality of $T_3\subseteq E(\FF_{q^3})$ is in the range of $q^2$ and the complexity of the DLP is $\mathcal{O}(q)$, that is, the square root of the group order (see \cite[Section 22.3.4.b]{handbook}). Hence, from the point of view of security, the degree three trace-zero subgroup of an elliptic curve defined over $\FF_q$ is comparable to the group of points of an elliptic curve over a ground field $\FF_p$, where $p$ is in the range of $q^2$.

On the other hand, Weil restriction of scalars allows us to regard $E(\FF_{q^n})$ as the set of $\FF_q$-rational points of a variety of dimension $n$ defined over $\mathbb{F}_q$, and $T_n$ as the set of $\FF_q$-rational points of a subvariety of dimension $n-1$. Hence one would like to be able to represent the elements of $T_n$ via $n-1$ $\FF_q$-coordinates, as opposed to the $n$ $\FF_q$-coordinates needed to represent an element of $E(\FF_{q^n})$. Optimal representations for the degree $n$ trace-zero subgroup of an elliptic curve have been proposed by Naumann in \cite{nau99} for $n=3$, Silverberg in \cite{sil05} and Cesena in \cite{ces08} for $n=3,5$, and Gorla-Masserier in~\cite{EM1} for small values of $n$ and in~\cite{EM2} for any $n$. Optimal coordinates for the degree $n$ trace-zero subgroup of a hyperelliptic curves of genus $g$ were proposed by Lange in \cite{langetzero} for $g=2$ and $n=3$, and by Gorla-Massierer in \cite{EM2} for any $g\geq 1$ and $n\geq 2$. 

In order to take full advantage of the optimal representation size for level of security in trace-zero subgroups, one needs efficient algorithms to perform arithmetic on the group elements represented in compressed coordinates. There are two natural ways to perform scalar multiplication in $T_n$: One can either compute scalar multiplication in $E(\FF_{q^n})$ and use compression and decompression algorithms to go back and forth between the usual coordinates in $E(\FF_{q^n})$ and the compressed coordinates in $T_n$, or compute scalar multiplication directly in compressed coordinates in $T_n$. 

The first approach is relatively straightforward: In all previously quoted work dealing with optimal representations in $T_n$, the authors provide compression and decompression algorithms. There is a wealth of knowledge on how to efficiently perform scalar multiplication on elliptic curves and, in addition, the Frobenius endomorphism $\varphi$ on the curve allows us to speed up scalar multiplication in $E(\FF_{q^n})$, as explained in \cite[Sections 15.1 and 15.2]{handbook}. Following this approach, computing scalar multiplication in $T_3$ is usually faster than in the group of rational points of a curve over a ground field of prime size in the range of $q^2$. Observe also that in $T_3$ scalar multiplications can be further sped up by using the relation $\varphi^2+\varphi+1=0$ involving the Frobenius endomorphism (see \cite[Section 15.3]{handbook}, \cite{ac07}, \cite{TrZeroAvanzi}, \cite{langetzero}, \cite{nau99}, \cite{Weim}). Using the same approach, one can also speed up the computation of the Miller function for the Tate pairing, in the context of pairing-based cryptography (see \cite{ces08}). 

The second approach is performing scalar multiplication in $T_n$ in the optimal compressed coordinates. To the extent of our knowledge, no such algorithm has been proposed yet. In this paper, we give an algorithm to perform scalar multiplication in the degree three trace-zero subgroup of an elliptic curve, in the representation proposed in \cite{EM2}.
Namely, let $E$ be an elliptic curve over $\FF_q$, whose degree three trace-zero subgroup $T_3$ is cyclic of prime order $p$. Our algorithm takes as input an integer $m$ modulo $p$ and the line through $P\in T_3$ and its Frobenius conjugates, and it returns the line through the point $mP$ and its Frobenius conjugates. Our algorithm has interesting similarities with the Montgomery ladder algorithm for computing scalar multiplication for elliptic curves, when the points are represented using their $x$-coordinate (see \cite{Mont} and \cite[Section 13.2.3.d]{handbook}). Moreover, our algorithm adapts the above mentioned strategy for exploiting the relation $\varphi^2+\varphi+1=0$ satisfied by the Frobenius endomorphism. Hence, we can maintain the advantages of such a strategy, even performing the operation directly in compressed coordinates. 

The paper is organized as follows. In Section~\ref{prelim} we establish the notations and some preliminaries on the degree three trace-zero subgroup of an elliptic curve. We also 
present some procedures for computation, that will be used in the subsequent algorithms. 
In Section~\ref{main} we present our algorithm for scalar multiplication. 
Subsection~\ref{subalg} contains a subalgorithm that will be called by the main algorithms, and a lemma which allows us to deal with special cases. 
In Subsection~\ref{montgomery} we propose a Montgomery-ladder-style algorithm which computes scalar multiplication in $T_3$. The algorithm makes use of the subalgorithm of Subsection~\ref{subalg}. In Subsection~\ref{combined} we exploits the properties of the Frobenius endomorphism to obtain an optimized version of the Montgomery-ladder-style algorithm of Subsection~\ref{montgomery}. The resulting algorithm efficiently computes scalar multiplication in $T_3$. In the Appendix we give the explicit formulas that we have computed and that we use for computation.

\section{Setting, notation, and formulas}\label{prelim}

\subsection{Preliminaries and notation}

Let $\F_q$ be a finite field of characteristic different from $2$ and $3$. 
Let $E$ be an elliptic curve defined over $\mathbb{F}_q$ by an equation in short Weierstrass form, i.e. $E$ is the zero-locus of a polynomial of the form $y^2-f(x)$, where $f(x)=x^3+Ax+B$ has no multiple roots and $A,B\in\FF_q$. 
Denote by $+$ the usual addition between points of $E$ and by $P_{\infty}$ the neutral element of $E$. 
For a field extension $\F_q \subseteq {\F}_{q^n}$, denote by $E(\mathbb{\mathbb{F}}_{q^n})$ the group of $\mathbb{F}_{q^n}$-rational points of $E$.  

Consider the Frobenius endomorphism on the group of $\FF_{q^3}$-rational points of $E$:
$$\varphi : E({\F}_{q^3}) \longrightarrow E({\F}_{q^3}), \ \ \ (x,y) \mapsto (x^q,y^q) \mbox{, } P_{\infty} \mapsto P_{\infty}.$$
The Frobenius endomorphism induces the trace endomorphism:
$$\mathrm{Tr} : E(\F_{q^3}) \longrightarrow E(\F_q), \ \ \ P \mapsto P + \varphi(P) + \varphi^{2}(P),$$
whose kernel is the trace zero subgroup $T_3$ of $E(\F_{q^3})$, i.e. 
$$T_3 = \{P \in E(\F_{q^3}) \mbox{ : } P+\varphi(P) +  \varphi^{2}(P)=P_{\infty}\}.$$
Let $P =(x_P,y_P)\in T_3\setminus\{P_\infty\}$ and denote by $h_P$ the equation of the line through $P$, $\varphi(P)$, $\varphi^2(P)$. Then 
\begin{equation}\label{hPmonic}
h_P = y-( \alpha_1x+\alpha_0)
\end{equation}
with $\alpha_1,\alpha_0 \in \mathbb{F}_q$. By~\cite[Corollary 4.2]{EM2}, $h_P$ of the form (\ref{hPmonic}) exists and is unique.
Notice moreover that $$h_{-P}(x,y)=-h_P(x,-y)=y+( \alpha_1x+\alpha_0).$$

Following \cite{EM2}, we represent an element $P\in T_3\setminus\{P_\infty\}$ via the coefficients $(\alpha_0,\alpha_1)$ of $h_P$. Such a representation is optimal in size, since $T_3$ is a variety of dimension $2$ over $\mathbb{F}_q$. Intuitively, optimality means that the number of coordinates is the least possible, see~\cite[Definition~2.7]{EM2} for the formal definition of an optimal representation. In this paper we give an algorithm to compute scalar multiplication in $T_3$ using the representation from \cite{EM2}. Scalar multiplication is the operation needed in most applications, e.g. in the Diffie-Hellman key agreement.

Notice that the representation that we use identifies each point with its Frobenius conjugates. As a consequence, addition in compressed coordinates is not well-defined, that is, $h_P$ and $h_Q$ do not determine $h_{P+Q}$. However, scalar multiplication is well-defined: Given the line $h_P=0$ and an integer $m$, the line 
$h_{mP}=0$ through $mP$ and its Frobenius conjugates is uniquely determined. Observe the analogy with the representation of points of $E$ via their $x$-coordinates: $m$ and the $x$-coordinate of a point $P\in E$ determine the $x$-coordinate of $mP$, however the $x$-coordinates of $P$ and $Q$ do not determine the $x$-coordinate of the point $P+Q$. 

In spite of the fact that one cannot compute $h_{P+Q}$ from $h_P$ and $h_Q$, one can compute the polynomial $S_{P,Q} \in \F_q[x,y]$ such that 
$$\divv(S_{P,Q})=\sum_{0\leq i,j \leq 2}(\varphi^i(P)+\varphi^j(Q)) - 9 P_{\infty}.$$
The polynomial $S_{P,Q}$ is unique up to multiplication by a nonzero constant and it is of the form
$$S_{P,Q} = (S_{P,Q})_1 + y(S_{P,Q})_2 = (a_4x^4+a_3x^3+a_2x^2+a_1x+a_0)+y(b_3x^3+b_2x^2+b_1x+b_0).$$
Notice that, if $P+Q,P+\varphi(Q),P+\varphi^2(Q)\neq P_\infty$, then 
\begin{equation}\label{spq} 
S_{P,Q}=h_{P+Q}h_{P+\varphi(Q)}h_{P+\varphi^2(Q)} \mod{y^2-f(x)}. 
\end{equation}

From $h_P$ and $S_{P,Q}$ one can compute the polynomials $$H_P:=f-(\alpha_1x+\alpha_0)^2, \ \Sigma_{P,Q}:=f(S_{P,Q})_2^2-(S_{P,Q})_1^2 \in\FF_q[x].$$
In the next lemma we collect a few useful facts.

\begin{lemma}\label{remHPQ}
Let $H_P=f-(\alpha_1x+\alpha_0)^2, \Sigma_{P,Q}=f(S_{P,Q})_2^2-(S_{P,Q})_1^2$. The following equalities hold, up to a nonzero constant:
\begin{enumerate}
\item $H_P=h_Ph_{-P} \mod{y^2-f(x)}$,
\item $H_P =(x-x_P)(x-x_P^{q})(x-x_P^{q^2})$,
\item $S_{-P,-Q}(x,y)=S_{P,Q}(x,-y)$,
\item $\Sigma_{P,Q}=S_{P,Q}S_{-P,-Q} \mod{y^2-f(x)}$,
\item $\Sigma_{P,Q}=\prod_{0\leq i,j\leq 2} (x-x_{\varphi^i(P)+\varphi^j(Q)})$.
\end{enumerate}
Moreover, the following are equivalent:
\begin{enumerate}[resume]
\item $(S_{P,Q})_2=0$, 
\item $b_3=0$,
\item $\varphi^i(P)+\varphi^j(Q)=P_\infty$ for some $i,j$,
\item $\divv(S_{P,Q})=(P-\varphi(P))+(\varphi(P)-P)+(P-\varphi^2(P))+(\varphi^2(P)-P)+(\varphi(P)-\varphi^2(P))+(\varphi^2(P)-\varphi(P))-6P_\infty$.
\end{enumerate}
\end{lemma}

\begin{proof}
{\em 1.} and {\em 2.} follow from~\cite[Corollary~4.2]{EM2}.\newline\noindent
{\em 3.} Observe that $\divv(S_{-P,-Q})=\sum_{0\leq i,j \leq 2}(-\varphi^i(P)-\varphi^j(Q)) - 9 P_{\infty}$, hence $$S_{-P,-Q}(x,y)=(S_{P,Q})_1(x)-y(S_{P,Q})_2(x)=S_{P,Q}(x,-y)$$ up to a nonzero constant.\newline\noindent
{\em 4.} By 3. $S_{P,Q}S_{-P,-Q}=(S_{P,Q})_1^2-y^2(S_{P,Q})_2^2=\Sigma_{P,Q}$, up to a nonzero constant.\newline\noindent
{\em 5.} By 4. $$\divv(\Sigma_{P,Q})=\divv(S_{P,Q})+\divv(S_{-P,-Q})=\sum_{0\leq i,j \leq 2}(\varphi^i(P)+\varphi^j(Q))+\sum_{0\leq i,j \leq 2}(-\varphi^i(P)-\varphi^j(Q)) - 18 P_{\infty},$$ hence $\Sigma_{P,Q}=\prod_{0\leq i,j\leq 2} (x-x_{\varphi^i(P)+\varphi^j(Q)})$ up to a nonzero constant.\newline\noindent
{\em 7.} $\Rightarrow$ {\em 8.} If $b_3=0$, then $\deg(\Sigma_{P,Q})\leq 8$, hence one of the sums $\varphi^i(P)+\varphi^j(Q)$ must be $P_\infty$.\newline\noindent
{\em 8.} $\Rightarrow$ {\em 9.} If $\varphi^i(P)+\varphi^j(Q)=P_\infty$ for some $i$ and $j$, then $S_{P,Q}=S_{\varphi^i(P),\varphi^j(Q)}=S_{\varphi^i(P),-\varphi^i(P)}=S_{P,-P}.$
Hence the zeroes of $S_{P,Q}$ on $E$ are $\pm(P-\varphi(P)),\pm(P-\varphi^2(P)),\pm(\varphi(P)-\varphi^2(P))$ and $P_\infty$, the latter with multiplicity six.\newline\noindent
{\em 9.} $\Rightarrow$ {\em 6.} 
Since the zeroes of $S_{P,Q}$ on $E$ are $\pm(P-\varphi(P)),\pm(P-\varphi^2(P)),\pm(\varphi(P)-\varphi^2(P))$ and $P_\infty$ with multiplicity six, then
$S_{P,Q}=(x-x_{P-\varphi(P)})(x-x_{P-\varphi^2(P)})(x-x_{\varphi(P)-\varphi^2(P)})\in\FF_q[x]$. Hence $(S_{P,Q})_2=0$.
\end{proof}

\subsection{Procedures for computing doubling and tripling formulas, and the coefficients of $S_{P,Q}$} 

In this subsection we describe two procedures which allow us to compute doubling and tripling formulas for the equation of a line, and the coefficients of the polynomial $S_{P,Q}$. More precisely:

\begin{itemize}
\item Following Procedure 1, we were able to write explicit formulas for the coefficients of $S_{P,Q}$ in terms of the coefficients of $h_P$ and $h_Q$ (see formulas (1) in the appendix) and for the coefficients of $h_{2P}$ in terms of the coefficients of $h_P$ (see formulas (2) in the appendix).
\item Following Procedure 2, we wrote explicit formulas for the coefficients of $h_{3P}$ in terms of the coefficients of $h_P$ (see formulas (3) in the appendix).
\end{itemize}

Moreover, in Proposition \ref{hPQ} we give a procedure to compute the coefficients of $h_{P+Q}$ in terms of the coefficients of $H_{P+Q}$ and $S_{P,Q}$. We assume that $(S_{P,Q})_2 \not = 0,H_{P+Q}$ and that $H_{P+Q}$ is irreducible over $\F_q[x]$ (i.e., that $P+Q\not\in E[3](\FF_q)$).

\begin{notation}
For Procedures 1 and 2, we let $\varphi^{i-1}(P)=P_i=(x_{P_i},y_{P_i})$, respectively $\varphi^{i-1}(Q))=Q_i=(x_{Q_i},y_{Q_i})$ for $i\in \{1,2,3\}$.
We denote by $e_1,e_2,e_3$ the symmetric polynomials in $x_{P_1},x_{P_2},x_{P_3}$ and by $s_1,s_2,s_3$ the symmetric polynomials in $x_{Q_1},x_{Q_2},x_{Q_3}$.
\end{notation}

\hrule
\begin{procedure} 
\em{Procedure to write formulas for the coefficients of $h_{2P}$ in terms of those of $h_P$
and for the coefficients of $S_{P,Q}$ in terms of those of $h_P$ and $h_Q$.\\
\hrule
\textbf{ }\\
\begin{footnotesize}1: \end{footnotesize}\textbf{for} $i \in \{1,2,3\}$
\hspace{1cm}\begin{footnotesize}$\triangleright$ $t_{i}=0$ tangent to $E$ in $P_i$, $t_i$ polynomial in the variables $x_{P_i},y_{P_i},x,y$\end{footnotesize}\\
\begin{footnotesize}2: \end{footnotesize}\mbox{ }\mbox{ }$t_{i}(x_{P_i},y_{P_i},x,y) \leftarrow
f'(x_{P_i})x - 2y_{P_i}y + (2y_{P_i}^2-f'(x_{P_i})x_{P_i})$\\
\begin{footnotesize}3: \end{footnotesize}\mbox{ }\mbox{ }\textbf{for} $j \in \{1,2,3\}$
\hspace{0.8cm}\begin{footnotesize}$\triangleright$ $r_{ij}=0$ line through $P_i$ and $Q_j$, $r_{ij}$ polynomial in the variables  $x_{P_i},y_{P_i},x_{Q_j},y_{Q_j},x,y$\end{footnotesize}\\
\begin{footnotesize}4: \end{footnotesize}\mbox{ }\mbox{ }\mbox{ }\mbox{ }$r_{ij}(x_{P_i},x_{Q_j},y_{P_i},y_{Q_j},x,y) \leftarrow
(y_{Q_j}-y_{P_i})x + (x_{P_i}-x_{Q_j})y + ((x_{Q_j}-x_{P_i})y_{P_i}+(y_{P_i}-y_{Q_j})x_{P_i})$\\
\begin{footnotesize}5: \end{footnotesize}\mbox{ }\mbox{ }\textbf{end for}\\
\begin{footnotesize}6: \end{footnotesize}\textbf{end for}\\
\begin{footnotesize}7: \end{footnotesize}
$T(x_{P_1},x_{P_2},x_{P_{3}},y_{P_1},y_{P_2},y_{P_{3}},x,y)\leftarrow \prod_{i=1}^{3}t_{i}$\\
\begin{footnotesize}8: \end{footnotesize}
$R(x_{P_1},x_{P_2},x_{P_{3}},y_{P_1},y_{P_2},y_{P_{3}},x_{Q_1},x_{Q_2},x_{Q_{3}},y_{Q_1},y_{Q_2},y_{Q_{3}},x,y)\leftarrow \prod_{1\leq i,j \leq 3}r_{i,j}$\\
\begin{footnotesize}9: \end{footnotesize}\textbf{for} $i \in \{1,2,3\}$\\
\begin{footnotesize}10: \end{footnotesize}\mbox{ }\mbox{ }
replace $y_{P_i}$ with $(\alpha_1x_{P_i}+\alpha_0)$ in $T$ and in $R$\\
\begin{footnotesize}11: \end{footnotesize}\mbox{ }\mbox{ }
replace $y_{Q_i}$ with $(\beta_1x_{Q_i}+\beta_0)$ in $R$\\
\begin{footnotesize}12: \end{footnotesize}\textbf{end for}\\
\begin{footnotesize}13: \end{footnotesize}\mbox{ }\mbox{ }write 
$T(x_{P_1},x_{P_2}, x_{P_{3}})$, $R(x_{P_1},x_{P_2}, x_{P_{3}})$ as polynomials in $e_1,e_2,e_{3}$\\
\begin{footnotesize}14: \end{footnotesize}\mbox{ }\mbox{ }write $R(x_{Q_1},x_{Q_2}, x_{Q_{3}})$ as a polynomial in $s_1,s_2,s_{3}$\\
\begin{footnotesize}15: \end{footnotesize}\mbox{ }\mbox{ }$E_1 \leftarrow \alpha_1^2$, $E_2 \leftarrow A-2\alpha_0\alpha_1$, $E_3 \leftarrow \alpha_0^2-B$\\
\begin{footnotesize}16: \end{footnotesize}\mbox{ }\mbox{ }$S_1 \leftarrow \beta_1^2$, $S_2 \leftarrow A-2\beta_0\beta_1$, $S_3 \leftarrow \beta_0^2-B$\\
\begin{footnotesize}17: \end{footnotesize}\textbf{for} $i \in \{1,2,3\}$\\
\begin{footnotesize}18: \end{footnotesize}\mbox{ }\mbox{ }replace $e_i$ with $E_i$ in $T$, $R$\\
\begin{footnotesize}19: \end{footnotesize}\mbox{ }\mbox{ }replace $s_i$ with $S_i$ in $R$\\
\begin{footnotesize}20: \end{footnotesize}\textbf{end for}\\
\begin{footnotesize}21: \end{footnotesize}
recover $h_{2P}$ via the equality (up to multiplication by a nonzero constant):
$$h_{2P}= T(x,-y)/(h_{-P}^2) \mod{y^2-f(x)}.$$ 
\begin{footnotesize}22: \end{footnotesize} 
recover $S_{P,Q}$ via the equality (up to multiplication by a nonzero constant): 
$$(S_{P,Q})_1(x)-y(S_{P,Q})_2(x) = R(x,y)/(h_P^3h_Q^3) \mod{y^2-f(x)}.$$
}
\end{procedure}
\hrule
\bigskip

\begin{theorem}\label{thformule}
Procedure 1 is correct.
\end{theorem}

\begin{proof}
We first prove that the formulas of Procedure 1 are correct when $h_P\not = y$ and $h_Q \not =  h_{\pm P}$. 
We regard $x_{P_1},x_{P_2},x_{P_3},x_{Q_1},x_{Q_2},x_{Q_3}$ as variables.
Since $h_P \not = y$, one has that $2P_i\not = P_{\infty}$ for $i \in \{1,2,3\}$, so $t_i(x_{P_i},y_{P_i},x,y)=0$ the equation defining the tangent to $E$ at $P_i$ is of the form 
given in line $2$ and $\divv(t_i) = P_i+P_i+(-2P_i) - 3P_{\infty}$.
Since $h_{Q} \not =  h_{\pm P}$, one has that $P_i \pm Q_j \not = P_{\infty}$ for $i, j\in \{1,2,3\}$. 
Then $r_{ij}(x_{P_i},x_{Q_j},y_{P_i},y_{Q_j},x,y)=0$, the equation of the line through $P_i$ and $Q_j$, 
is of the form given in line $4$ and $\divv(r_{ij})=P_i + Q_j + (-(P_i+Q_j)) - 3P_{\infty}.$
Let $T$ and  $R$ be as in lines $7$ and $8$ respectively. For $i \in \{1,2,3\}$, one has that 
$y_{P_i} =\alpha_1x_{P_i} + \alpha_0$ and $y_{Q_i}=\beta_1x_{Q_i} + \beta_0$ whence the correctness of lines $9-12$. 
Moreover, $T$, $R$ are symmetric polynomials in the variables $x_{P_1},x_{P_2},x_{P_{3}}$, and
$R$ is a symmetric polynomial in the variables $x_{Q_1},x_{Q_2},x_{Q_{3}}$. Hence they can be written as polynomial functions of $e_1,e_2,e_3$ and $s_1,s_2,s_3$. 
Correctness of lines $15$-$20$ follows from Lemma~\ref{remHPQ}.
Correctness of line $21$ follows from observing that
$$\divv(T) = \sum_{i=1}^{3}P_i + \sum_{i=1}^{3}P_i + \sum_{i=1}^{3}(-2P_i) -9P_{\infty} = 
2\divv(h_P) + \divv(h_{-2P})=\divv(h_P^2\cdot h_{-2P}),$$
hence $T = h_P^2\cdot h_{-2P} \mod{y^2-f(x)}$ up to multiplication by a nonzero constant.
Finally
$$\divv(R)= 3\sum_{i=1}^{3}P_i + 3\sum_{j=1}^{3}Q_j + \sum_{1\leq i,j \leq 3}(-(P_i+Q_j)) -27P_{\infty} =
\divv(h_P^{3}h_Q^{3}S_{P,Q}(x,-y)),$$ 
hence $R=h_P^{3}h_Q^{3}S_{P,Q}(x,-y) \mod{y^2-f(x)}$ up to multiplication by a nonzero constant, hence correctness of line $22$ follows. 
To conclude, one can directly check that the formulas computed in this way hold also in the case when $h_P=y$ or $h_Q=h_{\pm P}$. 
\end{proof}

\hrule
\begin{procedure} 
\em{
Procedure to write formulas for the coefficients of $h_{3P}$ in terms of those of $h_P$. \\
\hrule
\textbf{ }\\
\begin{footnotesize}1: \end{footnotesize}\textbf{for} $i \in \{1,2,3\}$
\hspace{2cm}\begin{footnotesize}$\triangleright$ doubling formulas for $P_i$ and $\ell_{i}=0$ line through $P_i$, $2P_i$  \end{footnotesize}\\
\mbox{ }\mbox{ }\mbox{ }\mbox{ }\begin{footnotesize}
$\triangleright$ $x_{2P_i}$ written as a rational function in the variables $x_{P_i}, y_{P_i}$ \end{footnotesize} \\
\begin{footnotesize}2: \end{footnotesize}\mbox{ }\mbox{ }$x_{2P_i}(x_{P_i},y_{P_i}) \leftarrow (f'(x_{P_i})/2y_{P_i})^2-2x_{P_i}$\\
\mbox{ }\mbox{ }\mbox{ }\mbox{ }\begin{footnotesize}
$\triangleright$ $y_{2P_i}$ written as a rational function in the variables $x_{P_i}, y_{P_i}$ \end{footnotesize} \\
\begin{footnotesize}3: \end{footnotesize}\mbox{ }\mbox{ }$y_{2P_i}(x_{P_i},y_{P_i}) \leftarrow (f'(x_{P_i})/2y_{P_i})(x_{P_i}-x_{2P_i})-y_{P_i}$\\
\mbox{ }\mbox{ }\mbox{ }\mbox{ }\begin{footnotesize}
$\triangleright$ $\ell_i$ written as a rational function in the variables $x_{P_i},y_{P_i},x,y$ \end{footnotesize} \\
\begin{footnotesize}4: \end{footnotesize}\mbox{ }\mbox{ }
$\ell_i(x_{P_i},y_{P_i},x,y) \leftarrow (y_{2P_i}-y_{P_i})x + (x_{P_i}-x_{2P_i})y + ((x_{2P_i}-x_{P_i})y_{P_i}+(y_{P_i}-y_{2P_i})x_{P_i})$\\
\begin{footnotesize}5: \end{footnotesize}\textbf{end for}\\
\begin{footnotesize}6: \end{footnotesize}
$L(x_{P_1},x_{P_2},x_{P_{3}},y_{P_1},y_{P_2},y_{P_{3}},x,y)\leftarrow \prod_{i=1}^{3}\ell_{i}$\\
\begin{footnotesize}7: \end{footnotesize}\textbf{for} $i \in \{1,2,3\}$\\
\begin{footnotesize}8: \end{footnotesize}\mbox{ }\mbox{ }
replace $y_{P_i}$ with $(\alpha_1x_{P_i}+\alpha_0)$ in $L$\\
\begin{footnotesize}9: \end{footnotesize}\textbf{end for}\\
\begin{footnotesize}10: \end{footnotesize}write 
$L(x_{P_1},x_{P_2}, x_{P_{3}})$ via the elementary symmetric polynomials $e_1,e_2,e_{3}$\\
\begin{footnotesize}11: \end{footnotesize}\mbox{ }\mbox{ }$E_1 \leftarrow \alpha_1^2$, $E_2 \leftarrow A-2\alpha_0\alpha_1$, $E_3 \leftarrow \alpha_0^2-B$\\
\begin{footnotesize}12: \end{footnotesize}\textbf{for} $i \in \{1,2,3\}$\\
\begin{footnotesize}13: \end{footnotesize}\mbox{ }\mbox{ }replace $e_i$ with $E_i$ in $L$\\
\begin{footnotesize}14: \end{footnotesize}\textbf{end for}\\
\begin{footnotesize}15: \end{footnotesize} Recover $h_{3P}$ using the formulas for $h_{2P}$ found with Procedure 1, together with the\\
\begin{footnotesize}\mbox{ }\mbox{ }\mbox{ }\mbox{ }\mbox{ } \end{footnotesize}equality (up to multiplication by a nonzero constant):
$$(h_{3P})= L(x,-y)/(h_{-P}h_{-2P}) \mod{y^2-f(x)}.$$
}
\end{procedure}
\hrule 
\bigskip

\begin{theorem}\label{corrproc2}
Procedure 2 is correct.
\end{theorem}
We omit the proof of Theorem~\ref{corrproc2}, since it is analogous to the proof of correctness for Procedure 1. 

We now want to compute $h_{P+Q}$ from $H_{P+Q}$ and $S_{P,Q}$. A straightforward way of doing this is computing the coefficients of $h_{P+Q}$ from those of $H_{P+Q}$ up to sign via the relations $w_2=-\gamma_1^2$, $w_1=A-2\gamma_0\gamma_1$, $w_0=B-\gamma_0^2$. One can then distinguish $h_{P+Q}=y-(\gamma_0+\gamma_1x)$ and $h_{-P-Q}=y+(\gamma_0+\gamma_1x)$, since $H_{P+Q}\mid (S_{P,Q})_1+(\gamma_0+\gamma_1x)(S_{P,Q})_2$. This however requires extracting a square root. 
The next proposition allows us to compute $h_{P+Q}$ from $H_{P+Q}$ and $S_{P,Q}$ more efficiently, by solving a simple linear system. 

\begin{proposition}\label{hPQ}
Suppose that $P+Q \not\in E[3](\FF_q)$, that $Q$ is not a Frobenius conjugate of $-P$ or $-2P$, and that $P$ is not a Frobenius conjugate of $-2Q$. 
Write $H_{P+Q}=x^3+w_2x^2+w_1x+w_0$ and $h_{P+Q} = y-(\gamma_1x+\gamma_0)$
with $\gamma_1,\gamma_0,w_2,w_1,w_0 \in \FF_q$.
Then $(\gamma_1,\gamma_0)$ is the unique solution of the linear system whose augmented matrix is
$$L(H_{P+Q},S_{P,Q}) = \left( \begin{matrix} w_0(w_2-b_2) & (b_0-w_0) & w_0a_3-a_4w_2w_0-a_0 \\ 
w_0(w_1-b_1) & (b_0w_2-w_0b_2) &   w_0a_2-a_4w_1w_0-a_0w_2 \\
w_0(w_0-b_0) & (b_0w_1-b_1w_0) & w_0a_1-a_4w_0^2-a_0w_1 
\end{matrix} \right).$$ 
\end{proposition}

\begin{proof} 
Using the fact that $H_{P+Q} | (S_{P,Q})_1+(\gamma_1x+\gamma_0)(S_{P,Q})_2$, 
a simple calculation shows that $(\gamma_1,\gamma_0)$ is a solution of the linear system with augmented matrix $L(H_{P+Q},S_{P,Q})$. 
Let us prove that the solution is unique. Let $(t_1,t_0)$ be a solution of the linear system with augmented matrix $L(H_{P+Q},S_{P,Q})$ and let $(x_0,y_0)\in T_3$ 
be one of the Frobenius conjugates of $P+Q$. Notice that, since $P+Q\not\in E[3](\FF_q)$, the three Frobenius conjugates are distinct. By construction,
$(S_{P,Q})_1(x_0)+(t_1x_0+t_0)(S_{P,Q})_2(x_0)=0$. We claim that $(S_{P,Q})_2(x_0)\not = 0$. In fact, if $(S_{P,Q})_2(x_0)=0$, then $(S_{P,Q})_2=H_{P+Q}$ and $H_{P+Q}\mid (S_{P,Q})_1$. In particular, $$0\leq \divv(S_{P,Q})-\divv(H_{P+Q})=\sum_{0\leq i,j\leq 2\\ i\neq j}\varphi^i(P)+\varphi^j(Q)-\sum_{i=0}^2\varphi^i(-P-Q),$$
hence $-P-Q=\varphi^i(P)+\varphi^j(Q)$ for some $i,j$ distinct. If $i,j\neq 0$, then $-\varphi^k(P)=P+\varphi^i(P)=-Q-\varphi^j(Q)=\varphi^h(Q)$ for some $h,k$, hence $P$ and $-Q$ are Frobenius conjugates. Similarly, $Q$ and $-2P$ are Frobenius conjugates if $i=0$ and $j\neq 0$, and $P$ and $-2Q$ are Frobenius conjugates if $i=0$ and $j\neq 0$. This concludes the proof of the claim. Since  $(S_{P,Q})_2(x_0)\not = 0$, then $y_0=t_1x_0+t_0$. 
Hence the line of equation $y-(t_1x+t_0)$ has three points in common with the line of equation $h_{P+Q}$. This implies that $t_1=\gamma_1$ and $t_0=\gamma_0$. 
\end{proof}

\begin{example}\label{exPQ}
Let $q = 1021$ and $\mathbb{F}_{q^3}=\mathbb{F}_q[\zeta]/(\zeta^3-5)$. Let $E$ be the elliptic curve over $\mathbb{F}_q$ of equation $y^2=x^3+230x+191$.
Let $P= (782\zeta^2 + 802\zeta + 45,979\zeta^2 + 299\zeta + 133)$, $Q=(466\zeta^2 + 528\zeta + 514 , 742\zeta^2 + 1016\zeta + 704) \in T_3$, 
with $h_P=y-(987x+642)$, $h_Q=y-(729x+705)$. 
Using the formulas in the appendix, we can compute:
$$h_{2P}=y-(1000x+280), \ \ h_{3P} = y-(646x+693),$$
$$S_{P,Q}= (823x^4 +  948x^3 + 709x^2 + 530x+ 741) + y(x^3 + + 782x^2 + 636x  + 100 ).$$
The matrix from Proposition \ref{hPQ} is:
$$L(H_{P+Q},S_{P,Q}) = \left( \begin{matrix} 
809 & 123 & 843 \\ 
568 & 823 & 755 \\
787 &  382 &  388
\end{matrix} \right).$$
Before we compute $L$, we compute $H_{P+Q}=x^3+880x^2+123x+998$ (in the next section we discuss how to compute $H_{P+Q}$). 
Solving the system associated to $L$ we find $h_{P+Q}=y-(65x+260)$.
\end{example}

\section{Scalar multiplication in $T_3$ using compressed coordinates}\label{main}

Throughout this section we assume that $T_3=\langle P \rangle$ is cyclic of order $p$, where $p$ is a prime of cryptographic size. 
Hence $\varphi(P)=sP$, with $s=(q-1)/(2+q-|E(\F_q)|) \mod{p}$, (see \cite[Section~15.3.1]{handbook}).
Let $m$ be an integer modulo $p$. In this section we develop an efficient algorithm to compute $h_{mP}$ given $m$ and $h_P$. 
In order to do this, in Subsection~\ref{subalg} we give a subalgorithm that we use within the main algorithm, as well as a lemma which helps us deal with special cases.
In Subsection~\ref{montgomery} we present a Montgomery-ladder-style algorithm that computes $h_{mP}$ from $m$ and $h_P$. 
Finally, in Subsection~\ref{combined} we apply the usual Frobenius endomorphism strategy to speed up our algorithm from Section~\ref{montgomery}. 
This gives our main algorithm to compute scalar multiplication in $T_3$ using compressed coordinates.

\subsection{Subalgorithm and special cases}\label{subalg}

Throughout this subsection $m$ is an integer $0<m<p$. Because of the doubling formulas in the Appendix, we may assume that $m$ is odd.

\begin{notation}\label{notret}
Let $m_1,m_2,n_1,n_2$ be integers such that $m_1+m_2=n_1+n_2=m$. For $i \in \{0,1,2\}$, 
let $h_i = h_{m_1P+\varphi^i(m_2P)}$, 
$H_i = H_{m_1P+\varphi^i(m_2P)}$,
$k_i=h_{n_1P+\varphi^i(n_2P)}$,
$K_i = K_{n_1P+\varphi^i(n_2P)}$. 
\end{notation}


Let $m_1,m_2,n_1,n_2$ be positive integers such that $m_1+m_2=n_1+n_2=m$ and suppose that we are given $h_{{m_1}P}, h_{{m_2}P},h_{{n_1}P},h_{{n_2}P}$. 
The subalgorithm computes $h_{mP}$ by applying the following strategy: Via the formulas found with Procedure 1, one can compute $$S_1:=S_{m_1P,m_2P}=S_{1,1}+yS_{1,2}$$ from $h_{m_1 P}, h_{m_2 P}$ and $$S_2:=S_{n_1 P,n_2 P}=S_{2,1}+yS_{2,2}$$ from $h_{n_1P}, h_{n_2P}$. 
Up to multiplying by a nonzero constant, $S_1=\prod_{i=0}^2h_i \mod{y^2-f(x)}$ and $S_2=\prod_{i=0}^2k_i \mod{y^2-f(x)}$, hence $S_1,S_2$ share the factor $h_0=k_0=h_{mP}$. 
By Lemma~\ref{remHPQ} $$H_{mP} | G:= \gcd{(fS_{1,2}^2-S_{1,1}^2,fS_{2,2}^2-S_{1,2}^2)}.$$ Moreover, if $m_1P+\varphi(m_2P)$ and $m_1P+\varphi(m_2P)$ are not Frobenius conjugates of $\pm(n_1P+\varphi(n_2P))$ or $\pm (n_1P+\varphi^2(n_2P))$, that is if $h_1,h_2\not\in\{k_1(x,y), k_2(x,y),-k_1(x,-y), -k_2(x,-y)\}$,
then $G=H_{mP}$. In this case, one can compute $h_{mP}$ from $G$ and $S_1$ (or from $G$ and $S_2$) by solving the linear system of Proposition \ref{hPQ}, provided that the assumptions of the proposition are satisfied. 

We now give the subalgorithm and we prove its correctness.\\
\medskip
\hrule
\begin{subalgorithm} \textbf{ }\\
{
\hrule
\textbf{ }\\
\textbf{Input}: The polynomials $h_{m_1P},h_{m_2P},h_{n_1P},h_{n_2P}$, such that $h_1, h_2 \not \in \{k_1,k_2\}$.\\
\textbf{Output} : $h_{mP}= y-(\gamma_1x+\gamma_0)$.\\
\hrule
\textbf{ }\\
\mbox{ }\mbox{ }\begin{footnotesize}1: \end{footnotesize}\textbf{if} $h_{m_1P}=h_{m_2P}$ \textbf{then return} $h_{-m_1P}$ \textbf{endif}\\
\mbox{ }\mbox{ }\begin{footnotesize}2: \end{footnotesize}\textbf{if} $h_{n_1P}=h_{n_2P}$ \textbf{then return} $h_{-n_1P}$ \textbf{endif}\\
\mbox{ }\mbox{ }\begin{footnotesize}3: \end{footnotesize}compute $S_1=S_{m_1P, m_2P}$ from $h_{m_1P}$, $h_{m_2P}$
\hspace{0.6cm}\begin{footnotesize}$\triangleright$ formulas $(1)$ in the appendix \end{footnotesize}\\
\mbox{ }\mbox{ }\begin{footnotesize}4: \end{footnotesize}compute $S_2=S_{n_1P,n_2P}$ from $h_{n_1P}$, $h_{n_2P}$\\
\mbox{ }\mbox{ }\begin{footnotesize}5: \end{footnotesize}\textbf{if} $h_{m_1P}(x,y)=-h_{m_2P}(x,-y)$ \textbf{then} \\
\mbox{ }\mbox{ }\begin{footnotesize}6: \end{footnotesize}\mbox{ }\mbox{ }
$W\leftarrow \monic(S_1)$\\
\mbox{ }\mbox{ }\begin{footnotesize}7: \end{footnotesize}\mbox{ }\mbox{ }
$L \leftarrow L(W,S_2)$ \hspace{5.2cm}\begin{footnotesize}$\triangleright$ see Proposition \ref{hPQ} \end{footnotesize}\\
\mbox{ }\mbox{ }\begin{footnotesize}8: \end{footnotesize}\mbox{ }\mbox{ }
compute $h=y-(\gamma_1x+\gamma_0)$ by solving the linear system associated to $L$\\
\mbox{ }\mbox{ }\begin{footnotesize}9: \end{footnotesize}\mbox{ }\mbox{ }\textbf{return} $h$\\
\mbox{ }\mbox{ }\begin{footnotesize}10: \end{footnotesize}\textbf{end if}\\
\mbox{ }\mbox{ }\begin{footnotesize}11: \end{footnotesize}\textbf{if} $h_{n_1P}(x,y)=-h_{n_2P}(x,-y)$ \textbf{then} \\
\mbox{ }\mbox{ }\begin{footnotesize}12: \end{footnotesize}\mbox{ }\mbox{ }
$W\leftarrow \monic(S_2)$\\
\mbox{ }\mbox{ }\begin{footnotesize}13: \end{footnotesize}\mbox{ }\mbox{ }
$L \leftarrow L(W,S_1)$ \hspace{5.2cm}\begin{footnotesize}$\triangleright$ see Proposition \ref{hPQ} \end{footnotesize}\\
\mbox{ }\mbox{ }\begin{footnotesize}14: \end{footnotesize}\mbox{ }\mbox{ }
compute $h=y-(\gamma_1x+\gamma_0)$ by solving the linear system associated to $L$\\
\mbox{ }\mbox{ }\begin{footnotesize}15: \end{footnotesize}\mbox{ }\mbox{ }\textbf{return} $h$\\
\mbox{ }\mbox{ }\begin{footnotesize}16: \end{footnotesize}\textbf{end if}\\
\mbox{ }\mbox{ }\begin{footnotesize}17: \end{footnotesize}$G \leftarrow \gcd(fS_{1,2}^2-S_{1,1}^2, fS_{2,2}^2-S_{2,1}^2)$\\
\mbox{ }\mbox{ }\begin{footnotesize}18: \end{footnotesize}decompose $G$ in irreducible factors in $\mathbb{F}_q[x]$\\
\mbox{ }\mbox{ }\begin{footnotesize}19: \end{footnotesize}$W_1,\cdots W_s \leftarrow$ monic distinct irreducible factors of $G$ of degree $3$\\
\mbox{ }\mbox{ }\begin{footnotesize}20: \end{footnotesize}\textbf{for} $j \in \{1,\cdots s\}$ \textbf{do}\\
\mbox{ }\mbox{ }\begin{footnotesize}21: \end{footnotesize}\mbox{ }\mbox{ }\mbox{ }$W \leftarrow W_j$\\
\mbox{ }\mbox{ }\begin{footnotesize}22: \end{footnotesize}\mbox{ }\mbox{ }\mbox{ }\textbf{if} $W \not = S_{1,2}$ \textbf{then}\\
\mbox{ }\mbox{ }\begin{footnotesize}23: \end{footnotesize}\mbox{ }\mbox{ }\mbox{ }\mbox{ }\mbox{ }
$L \leftarrow L(W,S_1)$ \hspace{5.2cm}\begin{footnotesize}$\triangleright$ see Proposition \ref{hPQ} \end{footnotesize}\\
\mbox{ }\mbox{ }\begin{footnotesize}24: \end{footnotesize}\mbox{ }\mbox{ }\mbox{ }\mbox{ }\mbox{ }
compute $h=y-(\gamma_1x+\gamma_0)$ by solving the linear system associated to $L$\\
\mbox{ }\mbox{ }\begin{footnotesize}25: \end{footnotesize}\mbox{ }\mbox{ }\mbox{ }\mbox{ }\mbox{ }
\textbf{if} $W | (\gamma_1 x + \gamma_0)S_{2,2}+S_{2,1}$ 
\textbf{then} return $h$ \\
\mbox{ }\mbox{ }\begin{footnotesize}26: \end{footnotesize}\mbox{ }\mbox{ }\mbox{ }\textbf{end if}\\
\mbox{ }\mbox{ }\begin{footnotesize}27: \end{footnotesize}\mbox{ }\mbox{ }\mbox{ }\textbf{else}
\hspace{7.2cm}\begin{footnotesize}$\triangleright$  $W=S_{1,2}$\end{footnotesize}\\
\mbox{ }\mbox{ }\begin{footnotesize}28: \end{footnotesize}\mbox{ }\mbox{ }\mbox{ }\mbox{ }\mbox{ }
$L \leftarrow  L(W,S_2)$ \hspace{5.2cm}\begin{footnotesize}$\triangleright$ see Proposition \ref{hPQ} \end{footnotesize}\\
\mbox{ }\mbox{ }\begin{footnotesize}29: \end{footnotesize}\mbox{ }\mbox{ }\mbox{ }\mbox{ }\mbox{ }
compute $h=y-(\gamma_1x+\gamma_0)$ by solving the linear system associated to $L$\\
\mbox{ }\mbox{ }\begin{footnotesize}30: \end{footnotesize}\mbox{ }\mbox{ }\mbox{ }\mbox{ }\mbox{ }
return $h$\\
\mbox{ }\mbox{ }\begin{footnotesize}31: \end{footnotesize}\mbox{ }\mbox{ }\mbox{ }\textbf{end if}\\
\mbox{ }\mbox{ }\begin{footnotesize}32: \end{footnotesize}\textbf{end for}
}
\end{subalgorithm}
\hrule
\medskip

\begin{theorem}\label{thsub}
Subalgorithm 1 is correct.
\end{theorem}

To prove the theorem we use the following.

\begin{remark}\label{rmkalg}
Since $T_3$ has prime order $p>3$, then $T_3\cap E[3](\FF_q)=\{P_\infty\}$. Hence $H_Q$ is irreducible over $\FF_q$ for every $Q\in T_3\setminus\{P_\infty\}$, in particular $H_{mP}$ is irreducible over $\FF_q[x]$ for every $0<m<p$. Moreover, $h_{mP}\not=h_{-mP}$, since, if this were the case, then $mP+\varphi^i(mP)=P_\infty$.
\end{remark}

\begin{proof}[Proof of Theorem \ref{thsub}]
If $h_{m_1P} = h_{m_2P}$ as in line $1$ of the subalgorithm, then 
$m_2P=\varphi^i(m_1P)$ for some $i\in \{0,1,2\}$. Since we assume that $m$ is odd, then $m_1\neq m_2$ and $m_1+m_2=m<p$, hence $i\neq 0$. 
Therefore $mP=(m_1+m_2)P=m_1(1+\varphi^i)(P)=-m_1\varphi^j(P)$ where $\{i,j\}=\{1,2\}$, and $i \not = j$. It follows that $h_{mP}=h_{-m_1P}$ and line $1$ is correct.
The same argument shows that, if $h_{n_1P}=h_{n_2P}$ as in line $2$ of the subalgorithm, then
$h_{mP}=h_{-n_1P}$, and line $2$ is correct. 

Correctness of lines $3,4$ follows from Theorem~\ref{thformule}. 

Up to multiplication by a nonzero constant, $S_1=h_{mP}h_1h_2$ and $S_2=h_{mP}k_1k_2 \mod{y^2-f(x)}$. Moreover, by Lemma \ref{remHPQ}, $fS_{1,2}^2-S_{1,1}^2=H_0H_1H_2$ and $fS_{2,2}^2-S_{2,1}^2=H_0K_1K_2$ up to multiplication by a nonzero constant. Suppose first that $h_{m_1P}=h_{-m_2P}$ as in line $5$. Then $S_1=h_{mP}(h_{-mP}) = H_{mP} \mod{y^2-f(x)}$ (up to multiplication by a nonzero constant). In addition, if $h_{m_1P}=h_{-m_2P}$, then $h_{n_1P}\not = h_{-n_2P}$. In fact, if $h_{n_1P}=h_{-n_2P}$, then $S_2=h_{mP}h_{-mP}=H_{mP}=S_1 \mod{y^2-f(x)}$ (up to multiplication by a nonzero constant), which is not possible since we are supposing $h_1,h_2 \not \in \{k_1,k_2\}$. The inequality $h_{n_1P}\not = h_{-n_2P}$ implies $S_{2,2}\not = 0$ by Lemma~\ref{remHPQ}. Moreover, by Remark \ref{rmkalg}, $H_{mP}$ is irreducible over $\FF_q[x]$. 
So, in order to apply Proposition \ref{hPQ} with $W=\monic(S_1)$ and $S_2$, it remains to prove that $H_{mP}\not = S_{2,2}$. Suppose this is not the case. Then $k_i=h_{-mP}$ for some $i \in \{0,1,2\}$. Since $h_{mP}\not= h_{-mP}$ by Remark \ref{rmkalg}, we have that $i \in \{1,2\}$ and $k_i=h_{-mP}=h_1$, which is not possible because $h_1, h_2 \not \in \{k_1,k_2\}$ by assumption. 
Hence one can apply Proposition \ref{hPQ} to $W=H_{mP}=\monic(S_1)$ and $S_2$, and correctness of lines $5-10$ follows. The proof of correctness of lines $11-16$ is analogous to that for lines $5-10$. 

From now on, we may assume that $h_{m_1P}\not =h_{-m_2P}$ and $h_{n_1P}\not = h_{-n_2P}$, which imply $S_{1,2}, S_{2,2} \not = 0$ by Lemma~\ref{remHPQ}.
Let $1 \leq s \leq 3$, $W_1,\ldots,W_s$ the monic distinct irreducible factors of degree $3$ over $\mathbb{F}_q[x]$ of $G=\gcd(fS_{1,2}^2-S_{1,1}^2,fS_{2,2}^2-S_{2,1}^2)$. By Remark \ref{rmkalg}, $H_0 \in \{W_1,\cdots, W_s\}$.
Moreover, for $W \in \{W_1,\ldots,W_s\}$, one has that $W=H_j$ for some $j \in \{0,1,2\}$.
Then, if $W \neq S_{1,2}$, one recovers $h=h_j$ from $W$ and $S_1$ by solving the linear system of Proposition \ref{hPQ}  (lines $22$-$24$ of the subalgorithm). 

We now consider line $25$. If $h=h_0=h_{mP}$, one has that $W | (\gamma_1 x +\gamma_0)S_{2,2}+S_{2,1}$. 
Else, $h \not = k_s$ for all $s \in \{0,1,2\}$, as $h_1,h_2 \not \in \{k_1,k_2\}$ by hypothesis. 
So $W \nmid(\gamma_1 x + \gamma_0)S_{2,2}+S_{2,1}$ by Proposition \ref{hPQ}, and line $25$ is correct.

Finally, suppose that $W = S_{1,2}$ as in line $26$. If $W \not = H_0$, one has that there exists $r \in \{1,2\}$ such that $h_j = -(h_r(x,-y))$. 
Moreover, there exists $s\in \{1,2\}$ such that $h_j=-(k_s(x,-y))$, since $W|G$ and $h_1,h_2 \not \in \{k_1,k_2\}$. 
Then $h_r=k_s$ with $r$, $s \in \{1,2\}$, that is not possible as $h_1,h_2 \not \in \{k_1,k_2\}$.
Hence $W=H_0$ and there exists $r \in \{1,2\}$ such that $h_{mP} \not = h_{rP} = h_{-mP}$,
from which $k_s \not = h_{-mP}$ for all $s\in \{0,1,2\}$, since $h_1,h_2 \not \in \{k_1,k_2\}$. 
So $W \not = S_{2,2}$, one recovers $h=h_{mP}$ from $W$ and $S_2$ by solving the linear system of Proposition \ref{hPQ}, and lines $26$-$30$ are correct. 
\end{proof}

We use the subalgorithm at each step of our Montgomery-ladder-style algorithm. We have two different types of input lines: The first is used in the general case, and the second for  special cases. 
\begin{enumerate}
\item[(a)] {\bf Input lines of type (a)}:
The subalgorithm computes $h_{mP}$ from $h_P$, $h_{(m-1)P}$, $h_{\frac{m-1}{2}P}$ and $h_{\frac{m+1}{2}P}$. 
The subalgorithm does not apply to a set $M$ of special values for $m$.
\item[(b)] {\bf Input lines of type (b)}: 
Let $R=\{(-3,-7),(-3,5),(3,-5),(3,7)\}$, $(r_1,r_2) \in R$. 
The subalgorithm computes $h_{mP}$ for $h_{r_iP}$, $h_{(m-r_i)P}$ for $i \in \{1,2\}$. 
The subalgorithm does not apply to a set $M_{(r_1,r_2)}$ of special values for $m$.
\end{enumerate}
In the next lemma we describe the sets $M$ and $M_{(r_1,r_2)}$. Moreover, we show that $M \cap (\bigcup_{(r_1,r_2) \in R}M_{(r_1,r_2)}) = \emptyset$. 
Therefore, one can compute $h_{mP}$ using the subalgorithm with input of type (a) if $m\not\in M$ and with input of type (b) if $m\in M$.

\begin{lemma}\label{m_ecc}
In the setting established above, one has the following:
\begin{enumerate}
\item $h_{ P + (m-1 )\varphi^i(P)} = h_{\frac{m-1}{2}P+\frac{m+1}{2}\varphi^j(P)} $ for some $i$, $j \in \{1,2\}$ if and only if $m \in M$, where
$$M=  
\left\{\frac{\pm 3}{2s+1},\frac{s-4}{3s},\frac{4s-1}{2s+1},\frac{s+5}{3(s+1)},\frac{4s+5}{2s+1} \mod{p}\right\}.$$
Hence Subalgorithm 1 correctly computes $h_{mP}$ from $h_P$, $h_{(m-1)P}$, $h_{\frac{m-1}{2}P}$ and $h_{\frac{m+1}{2}P}$ if $m \not \in M$.
\item Let $R=\{(-3,-7),(3,7),(-3,5),(3,-5)\}$, $(r_1,r_2) \in R$. 
Then $h_{ r_1P + (m-r_1)\varphi^i(P)} =h_{ r_2P + (m-r_2)\varphi^j(P)}$ for some $i,j \in \{1,2\}$ if and only if $m \in M_{(r_1,r_2)}$, where
\begin{itemize}
\item $M_{(3,7)}=\left\{\frac{17s+4}{2s+1}, \frac{-4s-17}{s-1}, \frac{10s+11}{2s+1},\frac{10s-1}{2s+1},\frac{4s-13}{-s-2},\frac{17s+13}{2s+1} \mod{(p)} \right\}$,
\item $M_{(-3,-7)}= \left\{-m \mod{(p)} \mid m \in M_{(3,7)}\right\}$,
\item $M_{(-3,5)}=\left\{\frac{7s+8}{2s+1},\frac{-8s-7}{s-1},\frac{2s+13}{2s+1},\frac{2s-11}{2s+1},\frac{8s+1}{-s-2},\frac{7s-1}{2s+1} \mod{(p)} \right\}$,
\item $M_{(3,-5)} = \left\{-m \mod{(p)} \mid m \in M_{(-3,5)}\right\}$.
\end{itemize}
Fix $(r_1,r_2) \in R$. Subalgorithm 1 correctly computes $h_{mP}$ from $h_{r_1P}, h_{r_2P}, h_{(m-r_1)P}, h_{(m-r_2)P}$ if $m \not \in M_{(r_1,r_2)}$.
\item One has that $M \cap (\bigcup_{(r_1,r_2) \in R}M_{(r_1,r_2)}) = \emptyset$. 
Hence, if Subalgorithm 1 cannot compute $h_{mP}$ with input of type (a), it can compute it with input of type (b). 
\end{enumerate}
\end{lemma}

\begin{proof}
By Theorem \ref{thsub}, and following Notation \ref{notret}, we have that Subalgorithm 1 correctly computes $h_{mP}$ from the input lines $h_{m_1P}=h_P$, $h_{m_2P}=h_{(m-1)P}$, $h_{n_1P}=h_{\frac{m-1}{2}P}$ and $h_{n_2P}=h_{\frac{m+1}{2}P}$ if $h_1,h_2\not \in \{k_1,k_2\}$, that is, if
$h_{P+(m-1)\varphi^i(P)}\not = h_{\frac{m-1}{2}P+\frac{m+1}{2}\varphi^j(P)}$ for all $i,j \in \{1,2\}$. 
We have that 
$$h_{P+(m-1)\varphi^i(P)} = h_{\frac{m-1}{2}P+\frac{m+1}{2}\varphi^j(P)} \mbox{ for some } i,j \in \{1,2\}$$ if and only if 
$$P+(m-1)\varphi^i(P)=\varphi^\ell\left(\frac{m-1}{2}P+\frac{m+1}{2}\varphi^j(P)\right) \mbox{ for some } i,j \in \{1,2\}, \ell \in \{0,1,2\}.$$ 
Since $\varphi(P)=sP$ and $P$ is of order $p$, the last equality is equivalent to 
\begin{equation}\label{equal}
1+(m-1)s^i=s^\ell\left(\frac{m-1}{2}+\frac{m+1}{2}s^j\right) \mod{p} \mbox{ for some } i,j \in \{1,2\}, \ell \in \{0,1,2\}.
\end{equation}
Moreover, $P \in {T}_3$, so $P+\varphi(P)+\varphi^2(P)=P_{\infty}$, hence
\begin{equation}\label{equazT3}
1+s+s^2=0 \mod{p},
\end{equation}
since $\varphi(P)=sP$ and $P$ has order $p$. 
From (\ref{equazT3}) one directly computes that (\ref{equal}) is equivalent to the statement that $m \in M$. Notice that all denominators in $M$ are nonzero modulo $p$, since (\ref{equazT3}) holds and $p \not= 2,3$. We have then proved part 1 of the lemma.

The proof for part 2 is analogous to that of part 1.

We now prove part 3. Suppose that $M \cap (\bigcup_{(r_1,r_2) \in R}M_{(r_1,r_2)}) \not = \emptyset$. 
One can check by direct computation that $as=b \mod{p}$ or $as=-b \mod{p}$ for some $a$ and $b$ such that $0<a,b\leq 60$ and $a\not=b$. If $as=b \mod{p}$, then from (\ref{equazT3}) one obtains that $a^2+ab+b^2 = 0 \mod{p}$, which is not possible since $0<a^2+ab+b^2 \ll p$. The case $as=-b \mod{p}$ can be treated similarly.
\end{proof}

\begin{remark} 
Lemma~\ref{m_ecc} is no longer true for small values of $p$. Consider e.g. the elliptic curve $y^2=x^3+5x+4$ over $\mathbb{F}_7$, with $p=31$ and $s=25$. 
We have $M\cap M_{(-3,-7)}=\{7,11,13\}\cap\{13,15\}=\{13\}\not = \emptyset$. 
\end{remark}

\begin{example}\label{exalg3}
Let $q = 1021$ and $\mathbb{F}_{q^3}=\mathbb{F}_q[\zeta]/(\zeta^3-5)$. We consider the same $E$ and $P$ as in Example \ref{exPQ}, i.e., we let $E$ be the elliptic curve over $\mathbb{F}_q$ of equation $y^2=x^3+230x+191$ and let $P= (782\zeta^2 + 802\zeta + 45,979\zeta^2 + 299\zeta + 133)$. Then $p=1021381$, $s=161217$, $M=\{161219, 322435, 322437, 465965 \}$.

We show how to compute $h_{5P}$ using Subalgorithm 1 with input of type (a). 
In Example \ref{exPQ} we computed $h_{2P}$ and $h_{3P}$. Using formulas $(1)$ and $(2)$ in the appendix, we compute $h_{4P}=y-(698x+155)$ from $h_{2P}$, 
$S_1=(524x^4  + 131x^3 + 826x^2  + 631x  + 160)+y(x^3 + 243x^2+ 651x+ 776)$ from $h_P$ and $h_{4P}$,
$S_2=(331x^4  + 653x^3  + 169x^2  + 259x  + 536) +y(x^3+ 570x^2+ 680x+ 578)$ from $h_{2P}$ and $h_{3P}$.
Then we compute $G=\gcd(fS_{1,2}^2-S_{1,1}^2, fS_{2,2}^2-S_{2,1}^2)= x^3 + 455x^2 + 81x+ 68$,
hence $G=H_{5P}$, and $H_{5P} \not = S_{1,2}$. 
So we obtain $h_{5P}=y-(736x+804)$ from $G$ and $S_1$ as in line $24$ of Subalgorithm 1. 

Similarly one can compute $h_{7P}=y-(112x+43)$ from $h_P$, $h_{6P}$, $h_{3P}$, $h_{4P}$.
\end{example}

The next two examples illustrate special cases of Subalgorithm 1.

\begin{example}\label{exalg3spec}
Let $E$ and $P$ be as in the previous example and let $m=337887$. One can check that 
$$P+(m-1)\varphi^2(P)=-\frac{m-1}{2}\varphi^2(P)-\frac{m+1}{2}\varphi(P).$$ If we try to compute $h_{mP}$  using Subalgorithm 1 with input of type (a), we first compute
$G=x^6 + 778x^5 + 86x^4 + 778x^3 + 599x^2 + 494x + 658$, which splits over $\mathbb{F}_q$ into two irreducible factors of degree $3$,
namely $W_1=x^3 + 11x^2 + 843x + 540$ and $W_2=x^3 + 767x^2 + 1016x + 5$.
From $W_1$ we recover $h_1=y-(166x+727)=0$ which is the line through $P+(m-1)\varphi^2(P)$, 
from $W_2$ we recover $h_2=y-(423x+57)=0$ which is the line through $mP$.
By checking the condition of line $25$ of the subalgorithm, we are able to decide that $h_{mP}=h_2$.
\end{example}

\begin{example}\label{exalg3spec2}
Let $q = 1021$ and $\mathbb{F}_{q^3}=\mathbb{F}_q[\zeta]/(\zeta^3-5)$. Let $E$ be the elliptic curve of equation $y^2=x^3+71x+529$ defined over $\FF_q$. Then $T_3$ is generated by $P=(853\zeta^2 + 995\zeta + 244 , 178\zeta^2 + 927\zeta + 959 )$, which has prime order $p=1009741$. Moreover $s=325960$ and $M_{(3,-5)} = \{32671,391027\}$.
Let $m=65339$. One can check that $mP=-3P-(m-3)\varphi^2(P)$. 
We compute $h_{mP}$ using Subalgorithm 1 with input of type (b), with $(r_1,r_2)=(3,-5)$. We obtain $G=S_{1,2}$, then we can compute $h_{mP}=y-(566x+37)$ from $G$ and $S_2$. 
\end{example}

\subsection{A first algorithm for scalar multiplication}\label{montgomery}

We now present our Montgomery-ladder style algorithm for scalar multiplication in its basic form.

\begin{notation}
Let $m$ be an integer with $0<m<p$. Let $m=\sum_{i=0}^{\ell - 1}m_i2^i$ be the binary representation of $m$, with $m_i \in \{0,1\}$ for all $i$, $\ell = \lceil \log_2m \rceil$ and $m_{\ell-1}=1$. Let $$k_{i}=\sum_{j=i}^{\ell-1}m_j2^{j-i}$$ for $i \in \{0,\cdots, \ell-1\}$.
Notice that $k_{0}=m$. Finally, let 
$$M=  
\left\{\frac{\pm 3}{2s+1},\frac{s-4}{3s},\frac{4s-1}{2s+1},\frac{s+5}{3(s+1)},\frac{4s+5}{2s+1} \mod{p}\right\}$$ and define $\mathcal{M}=M\cap(2\mathbb{Z}+1)$.
\end{notation}

{\bf General strategy of the algorithm.} Our algorithm takes $h_P$ and $m$ as input, and it returns $h_{mP}$ as output. 
It adopts the classical double-and-add strategy for scalar multiplication: It computes
$$u_{i}=h_{k_{i}P} \mbox{ and } v_{i}=h_{(k_{i}+1)P}$$
for decreasing values of $i$. At the end of the cycle, it outputs $u_{0}=h_{mP}$. 
In order to compute the polynomials $u_{i}$ and $v_{i}$, the algorithm uses the doubling formulas of the appendix and Subalgorithm 1 with input the polynomials that it has computed in the previous steps.

The proposition below gives recursive definitions for $u_{i}$ and $v_{i}$ Our algorithm applies this proposition to construct the polynomials $u_{i}$ and $v_{i}$ at each step $i$. 

\begin{notation}
Write ${\Subalg}(h_1,h_2,h_3,h_4)$, for the output of Subalgorithm 1 with input $h_1,h_2,h_3,h_4$. 
For any $Q \in T_3$, let $D(h_Q)=h_{2Q}$, where $h_{2Q}$ is computed from the coefficients of $h_Q$ via the doubling formulas from the appendix. 
Then $D^k(h_Q)=h_{2^kQ}$, where $h_{2^kQ}$ is computed from $h_Q$ via iteration of the doubling formulas from the appendix. 
\end{notation}

\begin{proposition}\label{propalgo}
For $i$ from $i=\ell-1$ down to $i=0$, recursively define $u_{i}$ and $v_{i}$ as follows.
\begin{itemize}
\item $u_{\ell-1}=h_P$, $v_{\ell-1}=h_{2P}$.
\item $u_{\ell-2}=h_{2P}$ and $v_{\ell-2}=h_{3P}$ if $m_{\ell-2}=0$,\\
$u_{\ell-2}=h_{3P}$ and $v_{\ell-2}=h_{4P}$ if $m_{\ell-2}=1$.
\item For $0\leq i \leq \ell-3$:\\
\begin{itemize}
\item (General case) if $k_{i}, k_{i}+1 \not \in \mathcal{M}$, let\\
$u_{i}=D(u_{i+1})$ and $v_{i}=\Subalg(h_P,D(u_{i+1}),u_{i+1},v_{i+1})$ if $m_i=0$,\\
$u_{i}=\Subalg(h_P,D(u_{i+1}),u_{i+1},v_{i+1})$ and $v_{i}=D(v_{i+1})$ if $m_i=1$.\\
\item (Special cases) if $k_{i}$ or $k_{i}+1 \in \mathcal{M}$:\\
\begin{itemize}
\item If $m_i=0$, let \\
$u_{i}=D(u_{i+1})$ and 
$v_{i}$=$\left\{\begin{tabular}{lll}
$\Subalg(h_{3P},D^2(u_{i+2}),h_{7P},D^3(u_{i+3}))$ & if & $m_{i+1}=m_{i+2}=1$, \\
$\Subalg(h_{3P},D^3(u_{i+3}),h_{-5P},D^3(v_{i+3}))$ & if & $m_{i+1}=1, m_{i+2}=0$, \\
$\Subalg(h_{-3P},D^3(v_{i+3}),h_{5P},D^3(u_{i+3}))$ & if & $m_{i+1}=0, m_{i+2}=1$, \\
$\Subalg(h_{-3P},D^2(v_{i+2}),h_{-7P},D^3(v_{i+3}))$ & if & $m_{i+1}=m_{i+2}=0$. \\
\end{tabular}\right.$
\item If $m_i=1$, let \\
$v_{i}=D(v_{i+1})$ and 
$u_{i}$=$\left\{\begin{tabular}{lll}
$\Subalg(h_{3P},D^2(u_{i+2}),h_{7P},D^3(u_{i+3}))$ & if & $m_{i+1}=m_{i+2}=1$, \\
$\Subalg(h_{3P},D^3(u_{i+3}),h_{-5P},D^3(v_{i+3}))$ & if & $m_{i+1}=1, m_{i+2}=0$, \\
$\Subalg(h_{-3P},D^3(v_{i+3}),h_{5P},D^3(u_{i+3}))$ & if & $m_{i+1}=0, m_{i+2}=1$, \\
$\Subalg(h_{-3P},D^2(v_{i+2}),h_{-7P},D^3(v_{i+3}))$ & if & $m_{i+1}=m_{i+2}=0$. \\
\end{tabular}\right.$
\end{itemize}\end{itemize}
\end{itemize}
Then $u_{i}=h_{k_{i}P}$ and $v_{i}=h_{(k_{i}+1)P}$, for all $i \in \{0,\cdots, \ell-1\}$. 
\end{proposition}
\begin{proof}
We proceed by induction on $i$. The thesis is easily verified for $i=\ell-1$ and $i=\ell-2$. 
Hence let $0\leq i \leq \ell-3$ and assume that the thesis holds for $j \in \{i+1,\cdots,\ell-1\}$. 
Suppose first that $k_{i}, k_{i}+1 \not \in \mathcal{M}$ and that $m_i=0$ (the proof for the case $m_i=1$ is analogous). 
Then $k_{i}=2(k_{i+1})$ and $u_{i}=D(u_{i+1})=h_{2k_{i+1}P}=h_{k_{i}P}$ by induction. Moreover, by induction we get
$$\Subalg(h_P,D(u_{i+1}),u_{i+1},v_{i+1})=\Subalg(h_P,h_{2k_{i+1}P},h_{k_{i+1}P},h_{(k_{i+1}+1)P})=$$
$$\Subalg\left(h_P,h_{k_iP},h_{\frac{k_i}{2}P},h_{\left(\frac{k_i}{2}+1\right)P}\right).$$
Since $k_{i}+1\not \in \mathcal{M}$, Subalgorithm 1 with input of type (a) correctly outputs $v_{i}=h_{(k_{i}+1)P}$.
Now suppose that $k_{i}$ or $k_{i}+1 \in \mathcal{M}$
and assume that $m_i=0$, $m_{i+1}=m_{i+2}=1$ (the proof for the other cases is analogous). 
If $k_{i}$ or $k_{i}+1 \in \mathcal{M}$, then $i < \ell-3$, since $5,7 \not \in \mathcal{M}$. 
Hence we already have computed the polynomials of the three previous steps $i+1$, $i+2$, $i+3$. 
Since $m_i=0$, we prove the thesis for $u_{i}$ as in the general case. On the other hand, $k_{i}+1 \in \mathcal{M}$ so we cannot define $v_{i}$
using Subalgorithm 1 with input of type (a), as we did before. 
However $k_{i}+1=3+4k_{i+2}=7+8k_{i+3}$, so by induction we get
$$\Subalg(h_{3P},D^2(u_{i+2}),h_{7P},D^3(u_{i+3}))=\Subalg(h_{3P},h_{4(k_{i+2})P},h_{7P},h_{8(k_{i+3})P})=$$
$$\Subalg(h_{3P},h_{(k_{i}-2)P},h_{7P},h_{(k_{i}-6)P}).$$
Moreover, since $k_{i}+1 \in \mathcal{M}$, then $k_{i}+1 \not \in M_{(3,7)}$ by Lemma \ref{m_ecc}, hence Subalgorithm 1 with input of type (b) correctly outputs 
$v_{i}=h_{(k_{i}+1)P}$. 
\end{proof}

\begin{remark}
If $k_{i}, k_{i}+1 \not \in \mathcal{M}$, at step $i$ one needs only the polynomials computed in the previous step in order to compute the polynomials $u_i,v_i$.
If $k_{i}$ or $k_{i}+1 \in \mathcal{M}$ one needs the polynomials computed in the steps $i+2$ and $i+3$ in order to compute them. Therefore:
\begin{itemize}
\item In our algorithm, the last three pairs of polynomials that have been computed are stored in a vector $L$, which is updated at each step of the cycle. 
\item The algorithm looks for the $i$'s for which $k_{i}$ or $k_{i}+1 \in \mathcal{M}$ at the start: For each $i \in \{0,\cdots,\ell-2\}$, it computes $k_{i}$ and $k_{i}+1$, and it adds $i$ to the list $S$ if $k_{i}$ or $k_{i} +1\in \mathcal{M}$. Hence, at each step $i$, we know whether we have to call Subalgorithm 1 with input of type (a) or of type (b), by simply checking if $i \in S$. 
\end{itemize}
\end{remark}

\hrule
\begin{algorithm}[Scalar multiplication in $T_3$]
\textbf{ }\\
\hrule
\textbf{ }\\
\textbf{Input} : $h_P$, $m$ an integer modulo $p$.\\
\textbf{Output} : $h_{mP}$.\\
\hrule
\textbf{ }\\
\mbox{ }\mbox{ }\begin{footnotesize}1 : \end{footnotesize}
$m \leftarrow \sum_{i=0}^{\ell - 1}m_i2^i$ binary expansion of $m$\\
\mbox{ }\mbox{ }\mbox{ }\mbox{ }\mbox{ }\mbox{ }\mbox{ }\begin{footnotesize}$\triangleright$ collection of the special steps \end{footnotesize}\\
\mbox{ }\mbox{ }\begin{footnotesize}2 : \end{footnotesize} 
$S \leftarrow \{i \in \{0,\cdots,\ell-2\} \mbox{ : } k_i \leftarrow \sum_{j=i}^{\ell-1}m_j2^{j-i} \in \mathcal{M} \mbox{ or }  k_i+1 \in \mathcal{M}\}$\\
\mbox{ }\mbox{ }\mbox{ }\mbox{ }\mbox{ }\mbox{ }\mbox{ }\begin{footnotesize}$\triangleright$  step $i=\ell-1$ \end{footnotesize}\\
\mbox{ }\mbox{ }\begin{footnotesize}3 : \end{footnotesize} 
$u \leftarrow h_P$, $v \leftarrow h_{2P}$, $L\leftarrow [(u,v)]$
\hspace{0.5cm}\begin{footnotesize}$\triangleright$ $L=[(u_{\ell-1},v_{\ell-1})]$ \end{footnotesize}\\
\mbox{ }\mbox{ }\begin{footnotesize}4 : \end{footnotesize} 
\textbf{if} $\ell - 1 = 0$ \textbf{then} return $u$ \textbf{end if}\\
\mbox{ }\mbox{ }\mbox{ }\mbox{ }\mbox{ }\mbox{ }\mbox{ }\begin{footnotesize}$\triangleright$ step $i=\ell-2$ \end{footnotesize}\\
\mbox{ }\mbox{ }\begin{footnotesize}5 : \end{footnotesize} 
\textbf{if } $m_{\ell-2} = 0$ \textbf{then} $u \leftarrow h_{2P}$, $v \leftarrow h_{3P}$ \textbf{else} $u \leftarrow h_{3P}$, $v \leftarrow h_{4P}$ \textbf{end if}\\
\mbox{ }\mbox{ }\begin{footnotesize}6 : \end{footnotesize} 
Append $(u,v)$ to $L$
\hspace{0.5cm}\begin{footnotesize}$\triangleright$ $L=[(u_{\ell-1},v_{\ell-1}),(u_{\ell-2},v_{\ell-2})]$ \end{footnotesize}\\
\mbox{ }\mbox{ }\begin{footnotesize}7 : \end{footnotesize} 
\textbf{if} $\ell - 2 = 0$ \textbf{then} return $u$ \textbf{end if}\\
\mbox{ }\mbox{ }\mbox{ }\mbox{ }\mbox{ }\mbox{ }\mbox{ }\begin{footnotesize}$\triangleright$ cycle for: steps from $i=\ell-3$ to $i=0$ \end{footnotesize}\\
\mbox{ }\mbox{ }\begin{footnotesize}8: \end{footnotesize}
\textbf{for} $i$ from $\ell -3$ down to $0$ \textbf{do}\\
\mbox{ }\mbox{ }\mbox{ }\mbox{ }\mbox{ }\mbox{ }\mbox{ }\begin{footnotesize}$\triangleright$ special cases \end{footnotesize}\\
\mbox{ }\mbox{ }\begin{footnotesize}9: \end{footnotesize}
\mbox{ }\mbox{ }\textbf{if} $i \in S$ \textbf{then}\\
\mbox{ }\mbox{ }\begin{footnotesize}10: \end{footnotesize}
\mbox{ }\mbox{ }\mbox{ }\mbox{ }\textbf{if} $m_{i+1} =1$ \textbf{then}\\
\mbox{ }\mbox{ }\begin{footnotesize}11: \end{footnotesize}
\mbox{ }\mbox{ }\mbox{ }\mbox{ }\mbox{ }\mbox{ }\textbf{if} $m_{i+2} =1$ \textbf{then}\\
\mbox{ }\mbox{ }\begin{footnotesize}12: \end{footnotesize}
\mbox{ }\mbox{ }\mbox{ }\mbox{ }\mbox{ }\mbox{ }\mbox{ }\mbox{ }$h_{exc} \leftarrow \Subalg(h_{3P},D^2(L[2][1]),h_{7P},D^3(L[1][1]))$\\
\mbox{ }\mbox{ }\begin{footnotesize}13: \end{footnotesize}
\mbox{ }\mbox{ }\mbox{ }\mbox{ }\mbox{ }\mbox{ }\textbf{else}
\hspace{2.5cm}\begin{footnotesize}$\triangleright$ $m_{i+1}=1$, $m_{i+2}=0$  \end{footnotesize}\\
\mbox{ }\mbox{ }\begin{footnotesize}14: \end{footnotesize}
\mbox{ }\mbox{ }\mbox{ }\mbox{ }\mbox{ }\mbox{ }\mbox{ }\mbox{ }$h_{exc} \leftarrow \Subalg(h_{3P},D^3(L[1][1]),h_{-5P},D^3(L[1][2]))$\\
\mbox{ }\mbox{ }\begin{footnotesize}15: \end{footnotesize}
\mbox{ }\mbox{ }\mbox{ }\mbox{ }\mbox{ }\mbox{ }\textbf{end if}\\
\mbox{ }\mbox{ }\begin{footnotesize}16: \end{footnotesize}
\mbox{ }\mbox{ }\mbox{ }\mbox{ }\textbf{else} 
\hspace{2.7cm}\begin{footnotesize}$\triangleright$ $m_{i+1}=0$\end{footnotesize}\\
\mbox{ }\mbox{ }\begin{footnotesize}17: \end{footnotesize}
\mbox{ }\mbox{ }\mbox{ }\mbox{ }\mbox{ }\mbox{ }\textbf{if} $m_{i+2} =1$ \textbf{then}\\
\mbox{ }\mbox{ }\begin{footnotesize}18: \end{footnotesize}
\mbox{ }\mbox{ }\mbox{ }\mbox{ }\mbox{ }\mbox{ }\mbox{ }\mbox{ }$h_{exc} \leftarrow \Subalg(h_{-3P},D^3(L[1][2]),h_{5P},D^3(L[1][1]))$\\
\mbox{ }\mbox{ }\begin{footnotesize}19: \end{footnotesize}
\mbox{ }\mbox{ }\mbox{ }\mbox{ }\mbox{ }\mbox{ }\textbf{else}
\hspace{2.5cm}\begin{footnotesize}$\triangleright$ $m_{i+1}=0$, $m_{i+2}=0$  \end{footnotesize}\\
\mbox{ }\mbox{ }\begin{footnotesize}20: \end{footnotesize}
\mbox{ }\mbox{ }\mbox{ }\mbox{ }\mbox{ }\mbox{ }\mbox{ }\mbox{ }$h_{exc} \leftarrow \Subalg(h_{-3P},D^2(L[2][2]),h_{-7P},D^3(L[1][2]))$\\
\mbox{ }\mbox{ }\begin{footnotesize}21: \end{footnotesize}
\mbox{ }\mbox{ }\mbox{ }\mbox{ }\mbox{ }\mbox{ }\textbf{end if}\\
\mbox{ }\mbox{ }\begin{footnotesize}22: \end{footnotesize}
\mbox{ }\mbox{ }\mbox{ }\mbox{ }\textbf{end if}\\
\mbox{ }\mbox{ }\begin{footnotesize}23: \end{footnotesize}
\mbox{ }\mbox{ }\textbf{if} $|L|=3$ \textbf{then} remove $L[1]$ from $L$ \textbf{end if} 
\hspace{0.5cm}\begin{footnotesize}$\triangleright$ $L=[(u_{i+2},v_{i+2}),(u_{i+1},v_{i+1})]$ \end{footnotesize}\\
\mbox{ }\mbox{ }\mbox{ }\mbox{ }\mbox{ }\mbox{ }\mbox{ }\begin{footnotesize}$\triangleright$ computation of $u$, $v$ at step $i$\end{footnotesize}\\
\mbox{ }\mbox{ }\begin{footnotesize}24: \end{footnotesize}
\mbox{ }\mbox{ }\textbf{if} $m_i =0$ \textbf{then}\\
\mbox{ }\mbox{ }\begin{footnotesize}25: \end{footnotesize}
\mbox{ }\mbox{ }\mbox{ }\mbox{ }$u\leftarrow D(L[2][1])$\\
\mbox{ }\mbox{ }\begin{footnotesize}26: \end{footnotesize}
\mbox{ }\mbox{ }\mbox{ }\mbox{ }\textbf{if }$i \in S$ \textbf{then}\\
\mbox{ }\mbox{ }\begin{footnotesize}27: \end{footnotesize}
\mbox{ }\mbox{ }\mbox{ }\mbox{ }\mbox{ }\mbox{ }
$v \leftarrow h_{exc}$\\
\mbox{ }\mbox{ }\begin{footnotesize}28: \end{footnotesize}
\mbox{ }\mbox{ }\mbox{ }\mbox{ }\textbf{else}\\
\mbox{ }\mbox{ }\begin{footnotesize}29: \end{footnotesize}
\mbox{ }\mbox{ }\mbox{ }\mbox{ }\mbox{ }\mbox{ }
$v \leftarrow \Subalg(h_{P},D(L[2][1]),L[2][1],L[2][2])$\\
\mbox{ }\mbox{ }\begin{footnotesize}30: \end{footnotesize}
\mbox{ }\mbox{ }\textbf{else}
\hspace{2.5cm}\begin{footnotesize}$\triangleright$ $m_i=1$ \end{footnotesize}\\
\mbox{ }\mbox{ }\begin{footnotesize}31: \end{footnotesize}
\mbox{ }\mbox{ }\mbox{ }\mbox{ }\textbf{if } $i \in S$ \textbf{then}\\
\mbox{ }\mbox{ }\begin{footnotesize}32: \end{footnotesize}
\mbox{ }\mbox{ }\mbox{ }\mbox{ }\mbox{ }\mbox{ }
$u \leftarrow h_{exc}$\\
\mbox{ }\mbox{ }\begin{footnotesize}33: \end{footnotesize}
\mbox{ }\mbox{ }\mbox{ }\mbox{ }\textbf{else}\\
\mbox{ }\mbox{ }\begin{footnotesize}34: \end{footnotesize}
\mbox{ }\mbox{ }\mbox{ }\mbox{ }\mbox{ }\mbox{ }
$u \leftarrow \Subalg(h_{P},D(L[2][1]),L[2][1],L[2][2])$\\
\mbox{ }\mbox{ }\begin{footnotesize}35: \end{footnotesize}
\mbox{ }\mbox{ }\mbox{ }\mbox{ }\textbf{end if}\\
\mbox{ }\mbox{ }\begin{footnotesize}36: \end{footnotesize}
\mbox{ }\mbox{ }\mbox{ }\mbox{ }$v \leftarrow D(L[2][2])$\\
\mbox{ }\mbox{ }\begin{footnotesize}37: \end{footnotesize}
\mbox{ }\mbox{ }\textbf{end if}\\
\mbox{ }\mbox{ }\begin{footnotesize}38: \end{footnotesize}
\mbox{ }\mbox { }Append $(u,v)$ to $L$
\hspace{0.5cm}\begin{footnotesize}$\triangleright$ $L=[(u_{i+2},v_{i+2}),(u_{i+1},v_{i+1}),(u_{i},v_{i})]$ \end{footnotesize}\\
\mbox{ }\mbox{ }\begin{footnotesize}39: \end{footnotesize}
\textbf{end for}\\
\mbox{ }\mbox{ }\begin{footnotesize}40: \end{footnotesize}
return $L[3][1]$\\
\hrule
\end{algorithm}

\begin{theorem}\label{thmonlad}
Algorithm 1 is correct.
\end{theorem}

\begin{proof}
Correctness of lines $3-7$ is easy to check. 
Notice that, at the beginning of the cycle at line $8$, the list $L$ is $L=[(u_{\ell-1},v_{\ell-1}),(u_{\ell-2},v_{\ell-2})]$.
Moreover, one has that $\ell-3 \not \in S$, since $5,7 \not \in \mathcal{M}$, so we do not need to check whether $\ell-3\in S$. 
Observe now that for each $i$ from $i=\ell-3$ down to $i=0$, the list $L$ at line $23$ is $L=[(u_{i+2},v_{i+2}),(u_{i+1},v_{i+1})]$, while at line $38$ the list is 
$L=[(u_{i+2},v_{i+2}),(u_{i+1},v_{i+1}),(u_{i},v_{i})]$.
Hence correctness follows from Proposition~\ref{propalgo}.
\end{proof}

We now give an example of computation of a multiplication by $m$ for which the algorithm runs into the special cases.

\begin{example}\label{exmont} 
Let $q = 1021$ and $\mathbb{F}_{q^3}=\mathbb{F}_q[\zeta]/(\zeta^3-5)$. Let $E$ and $P$ be as in Example \ref{exPQ} and Example and \ref{exalg3}, i.e., let $E$ be the elliptic curve over $\mathbb{F}_q$ of equation $y^2=x^3+230x+191$ and let $P= (782\zeta^2 + 802\zeta + 45,979\zeta^2 + 299\zeta + 133)$. Let $m=644875$, with binary representation 
$$m=2^{19}+2^{16}+2^{15}+2^{14}+2^{12}+2^{10}+2^9+2^8+2^3+2+1.$$
For $i$ from $19$ to $0$ the pairs $(k_{i},k_{i}+1)$ are 
$$(1,2), (2,3), (4,5), (9,10), (19,20), (39,40), (78,79), (157,158), (314, 315), (629,630),$$
$$(1259,1260), (2519,2520), (5038,5039), (10076,10077), (20152,20153), (40304,40305),$$
$$(80609,80610), (161218,161219), (322437,322438),(644875,644876).$$
Hence the set of the special cases is $S=\{2,1\}$ since $k_{2}+1=161219$, $k_{1}=322437 \in \mathcal{M}$. We compute 
$h_{mP}=y-(105x+587)$ using Algorithm 1. At step $i=2$ we compute $v=h_{exc}$ with $m_3=1$ and $m_4=0$ (line $14$ of the algorithm). 
At step $i=1$ we compute $u=h_{exc}$ with $m_2=0$ and $m_3=1$ (line $18$ of the algorithm).
\end{example}

\subsection{The optimized algorithm for scalar multiplication}\label{combined}

In this subsection, we optimize the Montgomery-ladder style algorithm given in the previous subsection and give the conclusive algorithm to perform scalar multiplication in $T_3$ in optimal coordinates. 

\begin{remark}\label{remmeno}
Let $m$ be an integer modulo $p$. If $m>\frac{p-1}{2}$, one can reduce the computation of multiplication by $m$ to the computation of multiplication by $m'=-m \mod{p}$, with $m'\leq \frac{p-1}{2}$. One does so by using the equality $h_{-P}(x,y)=-h_P(x,-y)$.
\end{remark}

\paragraph{Frobenius reduction.} We now discuss how the Frobenius endomorphism can be used to increase the efficiency of our Montgomery-ladder-style algorithm for scalar multiplication. 

This strategy was first proposed by Koblitz in \cite{kob} for special elliptic curves and it has been applied to the group of $\mathbb{F}_{q^r}$-rational divisor classes of a hyperelliptic curve defined over $\mathbb{F}_q$ for $r>1$, see\cite[Section 15.1]{handbook}. The idea is splitting the computation of multiplication by $m$ in the computations of  several multiplications by smaller scalars. Such computations can be done in parallel, to obtain a faster scalar multiplication algorithm (see \cite[Section 15.1.2.d]{handbook}). 
In trace-zero subgroups, such a strategy enjoys the benefit of the extra property of the Frobenius on the trace, so that the operation can be further sped up. Hence computation in $T_n$ in the usual coordinates is faster than in the entire group, as shown in \cite[Section 15.3]{handbook}, \cite{ac07},\cite{TrZeroAvanzi}, \cite{langetzero}, \cite{nau99}, \cite{Weim}.

We now adapt this strategy to our scalar multiplication algorithm.
Let $m$ be an integer modulo $p$. One can write $m=m_0+sm_1$, with $m_0,m_1 \in \mathcal{O}(q)=\mathcal{O}(\sqrt{p})$, see the discussion in~\cite[Section 15.3.2]{handbook}. In order to compute $h_{mP}$ given $m$ and $h_P$, we call Algorithm 1 three times with input $m_0$, $m_1$ and $m_0+m_1$ respectively, instead of calling Algorithm 1 once with input $m$. Notice that $m_0,m_1,m_0+m_1 \in \mathcal{O}(\sqrt{p})$, while $m\in \mathcal{O}(p)$. Hence one reduces computation of the multiplication by $m$ 
to the computation of at most three multiplications by integers of smaller size. 
Similarly to what we did in Algorithm 1, one needs to pay attention to the special cases where one cannot apply Subalgorithm 1.  

\begin{lemma}\label{excred}
Let $m$, $m_0$, $m_1$ be integers modulo $p$, with $m_0,m_1 \not = 0$. One has the following:
\begin{enumerate}
\item Subalgorithm 1 with input $h_P$, $h_{mP}$, $h_{(m+1)P}$, $h_{(s-1)P}$ correctly outputs $h_{(m+s)P}$ if $m\not \in \mathcal{A}_1$, where
$$\mathcal{A}_1=\left\{-2, s, \frac{-3(1+s)}{2+s},\frac{-3}{2+s},\frac{s+2}{s-1},\frac{-3}{2s+1} \mod{p}\right\}.$$
\item Subalgorithm 1 with input $h_{mP}$, $h_{-mP}$, $h_{(m+s)P}$, $h_{-(m+1)P}$ correctly outputs $h_{m(1-s)P}$ if $m \not \in \mathcal{A}_2$, where 
$$\mathcal{A}_2=\left\{ 1,s,\frac{s+2}{s-1},\frac{2s+1}{-3},\frac{1-s}{3s} \mod{p}\right\}.$$
\item Subalgorithm 1 with input $h_{m_0P}$, $h_{m_1P}$, $h_{(m_0+m_1)P}$, $h_{m_0(1-s)P}$ correctly outputs $h_{(m_0+sm_1)P}$ if $2m_0+m_1 \not = 0 \mod{p}$ and $s \not \in \mathcal{B}_1$, where
$$\begin{array}{r}
\mathcal{B}_1 =\left\{\left(\frac{3m_0+m_1}{m_1}\right)^{\pm 1},\left(\frac{m_1-m_0}{2m_0+m_1}\right)^{\pm 1},\frac{m_0+2m_1}{-(2m_0+m_1)},\frac{3m_0+2m_1}{-(3m_0+m_1)},\frac{2m_1}{-(3m_0+m_1)} \mod{p} : \right. \\
\left. 3m_0+m_1, 2m_0+m_1,m_1-m_0 \not = 0 \mod{p} \right\}.\end{array}$$ 
\item Subalgorithm 1 with input $h_{m_0P}$, $h_{m_1P}$, $h_{(m_0+m_1)P}$, $h_{m_1(s-1)P}$ correctly outputs $h_{(m_0+sm_1)P}$ if $m_0+2m_1 \not = 0 \mod{p}$ and $s \not \in \mathcal{B}_2$, where
$$\begin{array}{r}
\mathcal{B}_2 = \left\{\left(\frac{m_0+3m_1}{-(2m_0+3m_1)}\right)^{\pm 1},\left(\frac{m_0-m_1}{m_0+2m_1}\right)^{\pm 1},\frac{2m_0+3m_1}{-m_0},\frac{m_0+3m_1}{-2m_0} \mod{p} : \right.\\
\left. m_0+3m_1,2m_0+3m_1,m_0+2m_1,m_1-m_0 \not = 0 \mod{p} \right\}\end{array}.$$
\item Let $\Poly = \{t+1,t-1,t+2,t+3,
3t+1,
t^2+1,
t^2+t+1,
t^2+4t+2,
2t^2+t+1,$\\
\newline
$t^2-t-1,
2t^2+4t+1,t^2+4t+1,
t^2+2t+2,
t^2+3t+1,
t^2+t-1,
2t^2+2t+1,$\\
\newline
$t^2+3t+1,
t^2-2t-1,t^2+2t-1,
2t^2+3t-1,2t^2+3t+1\}\subseteq \mathbb{F}_p[t]$\\
\newline
and let $\mathcal{R}$ be the corresponding set of roots in $\mathbb{F}_p$:
$$\mathcal{R}=\{\alpha \in \mathbb{F}_p \mid  f(\alpha)=0 \mbox{ for some } f \in \Poly\}.$$
Then $s \in \mathcal{B}_1\cap \mathcal{B}_2$ if and only if $m_0=\alpha m_1$ for some $\alpha \in \mathcal{R}$. 
\end{enumerate}
\end{lemma}

\begin{proof}
Recall that Subalgorithm 1 requires the condition $h_1,h_2 \not \in \{k_1,k_2\}$ for the input lines, where we follow Notation \ref{notret}. The lemma then follows from Theorem \ref{thsub} by direct computation (the proof is analogous to that of Lemma \ref{m_ecc}). 
\end{proof}

\paragraph{Precomputation.} In order to apply Frobenius reduction to scalar multiplication, we need to be able to deal with the special cases of Lemma \ref{excred}. We chose to solve this problem by using Algorithm 1 to precompute the polynomials of the set
\begin{equation}\label{prec}
\mathcal{L}=\{h_{m(1-s)P} \mbox{ : } m \in \mathcal{A}_1\cup \mathcal{A}_2\} \cup \{ h_{(s+\alpha)P} \mbox{ : } \alpha \in \mathcal{R}\}.
\end{equation}
In order to compute the polynomials of the form $h_{m(1-s)P}$, we first compute $h_{(s-1)P} \in \mathcal{L}$, then call Algorithm 1 with input $h_{(1-s)P}$ and $m$. 

We are now ready to present our final algorithm for scalar multiplication in $T_3$. Recall that at the end of the cycle for in Algorithm 1, one has computed the pair $L[3]=(h_{mP},h_{(m+1)P})$. 

\begin{notation}
Write $Alg_1(h_P,m)$ for the pair $(h_{mP},h_{(m+1)P})$, computed with a modified version of Algorithm 1 that outputs the entire pair $L[3]$.
\end{notation}

\hrule
\begin{algorithm}[Scalar multiplication in $T_3$]
\textbf{ }\\
\hrule
\textbf{ }\\
\textbf{Input} : $h_P$, $m$ an integer modulo $p$.\\
\textbf{Output} : $h_{mP}$.\\
\hrule
\textbf{ }\\
\mbox{ }\mbox{ }\begin{footnotesize}1 : \end{footnotesize} $\mathcal{L} \leftarrow$ set \ref{prec} of precomputed lines\\
\mbox{ }\mbox{ }\begin{footnotesize}2 : \end{footnotesize} \textbf{if} $m > \frac{p-1}{2}$ \textbf{then} $\overline{m} \leftarrow -m \mod{p}$ \textbf{else} $\overline{m} \leftarrow m$ \textbf{end if}\\
\mbox{ }\mbox{ }\begin{footnotesize}3 : \end{footnotesize} $\overline{m} \leftarrow m_0+sm_1$\\
\mbox{ }\mbox{ }\begin{footnotesize}4 : \end{footnotesize} \textbf{if} $m_0=0$ \textbf{then} $h \leftarrow Alg_1(h_P,m_1)[1]$\\
\mbox{ }\mbox{ }\begin{footnotesize}5 : \end{footnotesize} \textbf{else if} $m_1=0$ \textbf{then} $h \leftarrow Alg_1(h_P,m_0)[1]$\\
\mbox{ }\mbox{ }\begin{footnotesize}6 : \end{footnotesize} \textbf{else}
\hspace{2.5cm}\begin{footnotesize}$\triangleright$ $m_0,m_1\not = 0$  \end{footnotesize}\\
\mbox{ }\mbox{ }\begin{footnotesize}7 : \end{footnotesize} \mbox{ }\mbox{ }\textbf{if} $s\in \mathcal{B}_1\cap\mathcal{B}_2$ \textbf{then} 
\hspace{2.5cm}\begin{footnotesize}$\triangleright$ $\overline{m}=m_1(s+\alpha)$ for some $\alpha \in \mathcal{R}$  \end{footnotesize}\\
\mbox{ }\mbox{ }\begin{footnotesize}8 : \end{footnotesize} \mbox{ }\mbox{ }\mbox{ }\mbox{ }$h\leftarrow Alg_1(h_{(s+\alpha)P},m_1)[1]$
\hspace{1.35cm}\begin{footnotesize}$\triangleright$ $h_{(s+\alpha)P} \in \mathcal{L}$  \end{footnotesize}\\
\mbox{ }\mbox{ }\begin{footnotesize}9 : \end{footnotesize} \mbox{ }\mbox{ }\textbf{else}
\hspace{3.2cm}\begin{footnotesize}$\triangleright$ $s \not \in \mathcal{B}_1 \cap \mathcal{B}_2$ \end{footnotesize}\\
\mbox{ }\mbox{ }\begin{footnotesize}10 : \end{footnotesize} \mbox{ }\mbox{ }\mbox{ } $h_{m_0P}\leftarrow Alg_1(h_P,m_0)[1]$\\
\mbox{ }\mbox{ }\begin{footnotesize}11 : \end{footnotesize} \mbox{ }\mbox{ }\mbox{ } $h_{m_1P}\leftarrow Alg_1(h_P,m_1)[1]$\\
\mbox{ }\mbox{ }\begin{footnotesize}12 : \end{footnotesize} \mbox{ }\mbox{ }\mbox{ } $h_{(m_0+m_1)P}\leftarrow Alg_1(h_P,m_0+m_1)[1]$\\
\mbox{ }\mbox{ }\begin{footnotesize}13 : \end{footnotesize} \mbox{ }\mbox{ }\mbox{ } \textbf{if} $s \not \in \mathcal{B}_1$ and $2m_0+m_1 \not = 0 \mod{p}$ \textbf{then} 
\hspace{0.75cm}\begin{footnotesize}$\triangleright$ Compute $h_{(m_0+sm_1)P}$ from $h_{m_0P}, h_{m_1P}, h_{(m_0+m_1)P}, h_{m_0(1-s)P}$  \end{footnotesize}\\
\mbox{ }\mbox{ }\begin{footnotesize}14 : \end{footnotesize} \mbox{ }\mbox{ }\mbox{ }\mbox{ }\mbox{ }\mbox{ } \textbf{if} $m_0 \not \in \mathcal{A}_1\cup \mathcal{A}_2$ \textbf{then}\\
\mbox{ }\mbox{ }\begin{footnotesize}15 : \end{footnotesize} \mbox{ }\mbox{ }\mbox{ }\mbox{ }\mbox{ }\mbox{ }\mbox{ }\mbox{ } $h_{(m_0+1)P} \leftarrow Alg_1(h_P,m_0)[2]$\\
\mbox{ }\mbox{ }\begin{footnotesize}16 : \end{footnotesize} \mbox{ }\mbox{ }\mbox{ }\mbox{ }\mbox{ }\mbox{ }\mbox{ }\mbox{ } $h_{(m_0+s)P} \leftarrow \Subalg(h_P,h_{m_0P}, h_{(m_0+1)P},h_{(s-1)P})$\\
\mbox{ }\mbox{ }\begin{footnotesize}17 : \end{footnotesize} \mbox{ }\mbox{ }\mbox{ }\mbox{ }\mbox{ }\mbox{ }\mbox{ }\mbox{ } $h_{m_0(1-s)P} \leftarrow \Subalg(h_{m_0P},h_{-m_0P}, h_{-(m_0+1)P},h_{(m_0+s)P})$\\
\mbox{ }\mbox{ }\begin{footnotesize}18 : \end{footnotesize} \mbox{ }\mbox{ }\mbox{ }\mbox{ }\mbox{ }\mbox{ } \textbf{end if}\\
\mbox{ }\mbox{ }\begin{footnotesize}19 : \end{footnotesize} \mbox{ }\mbox{ }\mbox{ }\mbox{ }\mbox{ }\mbox{ } $h \leftarrow \Subalg(h_{m_0P},h_{m_1P}, h_{(m_0+m_1)P},h_{m_0(1-s)P})$\\
\mbox{ }\mbox{ }\begin{footnotesize}20 : \end{footnotesize} \mbox{ }\mbox{ }\mbox{ } \textbf{else}
\hspace{1.25cm}\begin{footnotesize}$\triangleright$ $s \not \in \mathcal{B}_2$ and $m_0+2m_1 \not = 0 \mod{p}$: Compute $h_{(m_0+sm_1)P}$ from $h_{m_0P}, h_{m_1P}, h_{(m_0+m_1)P}, h_{m_1(s-1)P}$  \end{footnotesize}\\
\mbox{ }\mbox{ }\begin{footnotesize}21 : \end{footnotesize} \mbox{ }\mbox{ }\mbox{ }\mbox{ }\mbox{ }\mbox{ } \textbf{if} $m_1 \not \in \mathcal{A}_1\cup \mathcal{A}_2$ \textbf{then}\\
\mbox{ }\mbox{ }\begin{footnotesize}22 : \end{footnotesize} \mbox{ }\mbox{ }\mbox{ }\mbox{ }\mbox{ }\mbox{ }\mbox{ }\mbox{ } $h_{(m_1+1)P} \leftarrow Alg_1(h_P,m_1)[2]$\\
\mbox{ }\mbox{ }\begin{footnotesize}23 : \end{footnotesize} \mbox{ }\mbox{ }\mbox{ }\mbox{ }\mbox{ }\mbox{ }\mbox{ }\mbox{ } $h_{(m_1+s)P} \leftarrow \Subalg(h_P,h_{m_1P}, h_{(m_1+1)P},h_{(s-1)P})$\\
\mbox{ }\mbox{ }\begin{footnotesize}24 : \end{footnotesize} \mbox{ }\mbox{ }\mbox{ }\mbox{ }\mbox{ }\mbox{ }\mbox{ }\mbox{ } $h_{m_1(1-s)P} \leftarrow \Subalg(h_{m_1P},h_{-m_1P}, h_{-(m_1+1)P},h_{(m_1+s)P})$\\
\mbox{ }\mbox{ }\begin{footnotesize}25 : \end{footnotesize} \mbox{ }\mbox{ }\mbox{ }\mbox{ }\mbox{ }\mbox{ } \textbf{end if}\\
\mbox{ }\mbox{ }\begin{footnotesize}26 : \end{footnotesize} \mbox{ }\mbox{ }\mbox{ }\mbox{ }\mbox{ }\mbox{ } $h \leftarrow \Subalg(h_{m_0P},h_{m_1P}, h_{(m_0+m_1)P},h_{m_1(s-1)P})$\\
\mbox{ }\mbox{ }\begin{footnotesize}27 : \end{footnotesize} \mbox{ }\mbox{ }\mbox{ } \textbf{end if}\\
\mbox{ }\mbox{ }\begin{footnotesize}28 : \end{footnotesize} \mbox{ }\mbox{ }\textbf{if} $m > \frac{p-1}{2}$ \textbf{then return} $-h(x,-y)$ \textbf{else return} $h$ \textbf{end if}\\
\hrule
\end{algorithm}

\begin{theorem}
Algorithm 2 is correct. 
\end{theorem}

\begin{proof} 
Let $\overline{m}$ be as in line $3$ of the algorithm. If $m_0=0$ as in line $4$, or $m_1=0$ as in line $5$, then $h=h_{\overline{m}P}$ by Theorem \ref{thmonlad}. 

Assume now that $m_0, m_1 \not =0$, as in line $6$. If $s \in \mathcal{B}_1\cap \mathcal{B}_2$ as in line $7$, then by Lemma \ref{excred}.5 $\overline{m}=m_1(s+\alpha)$ for some $\alpha \in \mathcal{R}$. In addition, $h_{(s+\alpha)P} \in \mathcal{L}$, where $\mathcal{L}$ is the set of precomputed polynomials of line $1$, defined in (\ref{prec}). Hence, by Theorem \ref{thmonlad}, one can compute $h=h_{\overline{m}P}$ as in line $8$ of  the algorithm.

Now consider the case in which $s \not \in \mathcal{B}_1 \cap \mathcal{B}_2$, as in line $9$ of the algorithm. Correctness of lines $10$, $11$ and $12$ follows from Theorem \ref{thmonlad}.

In line $13$ we have $s \not \in \mathcal{B}_1$ and $2m_0+m_1\not = 0 \mod{p}$. Then, by Lemma \ref{excred}.3, one can compute $h_{\overline{m}P}=h_{(m_0+sm_1)P}$ using Subalgorithm 1 with input lines $h_{m_0P}$, $h_{m_1P}$, $h_{(m_0+m_1)P}$ and $h_{m_0(1-s)P}$. 
We have already computed $h_{m_0P}$, $h_{m_1P}$, $h_{(m_0+m_1)P}$ in lines $10-12$. 
Hence,  in order to be able to compute $h_{\overline{m}P}$ with Subalgorithm 1, we still need to compute $h_{m_0(1-s)P}$, see also Lemma \ref{excred}.3. 

If $m_0 \not \in \mathcal{A}_1 \cup \mathcal{A}_2$ as in line $14$, then one computes $h_{m_0(1-s)P}$ as in lines $15-17$ of the Algorithm, by Theorem \ref{thsub}, Theorem \ref{thmonlad} and Lemma \ref{excred}, points 1 and 2.
If $m_0 \in \mathcal{A}_1\cup \mathcal{A}_2$, then we cannot compute the polynomial $h_{m_0(1-s)P}$ as we do in lines $15-17$ of the algorithm. Nevertheless, in this case, $h_{m_0(1-s)P}$ belongs to the set $\mathcal{L}$ of precomputed polynomials, by construction of $\mathcal{L}$. 
Therefore, in both cases Subalgorithm 1 in line $19$ correctly computes $h=h_{\overline{m}P}$, by Theorem \ref{thsub}. 

Now consider lines $20-26$ of the algorithm. We have either $s \in \mathcal{B}_1$ or $2m_0+m_1 = 0 \mod{p}$. 
Suppose first that $s \in \mathcal{B}_1$. Then $s \not \in \mathcal{B}_2$ , since $s \not \in \mathcal{B}_1\cap\mathcal{B}_2$.
Moreover, if $s \in \mathcal{B}_1$ then $m_0+2m_1 \not = 0 \mod{p}$. 
In fact, one can check by direct computation that $s \in \mathcal{B}_1$ and $m_0+2m_1 = 0 \mod{p}$ implies $s \in \{0,-5^{\pm 1},-3,-1,-4/5,2/5 \mod{p}\}$, since $m_0,m_1 \not = 0 \mod{p}$, which contradicts the equality $s^2+s+1=0 \mod{p}$. 
Now suppose that $2m_0+m_1 = 0 \mod{p}$. By the same arguments as above, one has that $2m_0+m_1 = 0 \mod{p}$ implies $s \not \in \mathcal{B}_2$ and $m_0+2m_1 \not = 0 \mod{p}$.
Hence, in both cases considered in line $20$, we have that $s \not \in \mathcal{B}_2$ and $m_0+2m_1\not = 0 \mod{p}$, and one can compute $h_{(m_0+sm_1)P}$ as in line $26$ by Lemma \ref{excred}, 4.

Similar arguments show that lines $21-25$ of the algorithm are correct, so Subalgorithm 1 at line $26$ correctly outputs $h=h_{\overline{m}P}$.
From line $2$, we have that $h=h_{\overline{m}P}=h_{-mP}$ if $m>\frac{p-1}{2}$, and $h=h_{\overline{m}P}=h_{mP}$ otherwise. Hence the algorithm correctly outputs $h_{mP}$ in line $28$ by Remark \ref{remmeno}. 
\end{proof}

\begin{remark}
The aim of Algorithm 2 is showing how to apply Frobenius reduction in order to speed up our scalar multiplication algorithm. 
However, further optimizations are possible. For example, one can introduce variations of Subalgorithm 1 in order to reduce the number of precomputed lines. 
\end{remark}

In conclusion, we give an example of optimized computation following with Algorithm 2. 

\begin{example}
Let $q = 1021$ and $\mathbb{F}_{q^3}=\mathbb{F}_q[\zeta]/(\zeta^3-5)$. Let $E$ and $P$ be as in Example \ref{exPQ}, Example \ref{exalg3}, and Example \ref{exmont}, i.e., let $E$ be the elliptic curve over $\mathbb{F}_q$ of equation $y^2=x^3+230x+191$ and let $P= (782\zeta^2 + 802\zeta + 45,979\zeta^2 + 299\zeta + 133)$. 
Write $m=483925=m_0+sm_1$, where $m_0=274$ and $m_1=3$. 

Algorithm 1 computes $h_{mP}$ by calling Subalgorithm 1 seventeen times with input $h_{m_1P},h_{m_2P},h_{n_1P},$ $h_{n_2P}$ for the following values of $(m_1,m_2,n_1,n_2)$: 
$$(1,6,3,4), (1,14,7,8), (1,28,14,15),(1,58,29,30), (1,118,59,60), (1,236,118,119),$$ 
$$(1,472,236,237),(1,944,472,473), (1,1890,945,946),(1,3780,1890,1891),$$
$$(1,7560,3780,3781), (1,15122,7561,7562),(1,30244,15122,15123),(1,60490,30245,30246),$$ 
$$(1,120980,60490,60491),(1,241962,120981,120982),(1,483924,241962,241963).$$ 

Performing the same computation with Algorithm 2, one has that $$s \not \in\mathcal{B}_1=\{275,757679,717376,508804,304004,263701,527404\},\  2m_0+m_1 \not = 0 \mod{p}$$ and $$m_0 \not \in \mathcal{A}_1\cup \mathcal{A}_2=\{1021379,860162,161216,860163,322435,161217,232982,627181\}.$$
Hence, after computing $h_{m_0P}$, $h_{(m_0+1)P}$, $h_{m_1P}$, $h_{(m_0+m_1)P}$, Algorithm 2 calls Subalgorithm 1 three times (in lines $16$, $17$ and $19$) in order to compute $h_{mP}$.  To compute $h_{m_0P}$ and $h_{(m_0+1)P}$, Algorithm 1 calls Subalgorithm 1 with input $h_{m_1P},h_{m_2P},h_{n_1P},h_{n_2P}$ for the following values of $(m_1,m_2,n_1,n_2)$:
$$(1,4,2,3), (1,8,4,5), (1,16,8,9),(1,34,17,18),(1,68,34,35), (1,136,68,69),(1,274,137,138).$$ 
To compute $h_{(m_0+m_1)P}$, Algorithm 1 calls Subalgorithm 1 with input $h_{m_1P},h_{m_2P},h_{n_1P},h_{n_2P}$ for the following values of $(m_1,m_2,n_1,n_2)$:
$$(1,4,2,3), (1,8,4,5), (1,16,8,9),(1,34,17,18),(1,68,34,35), (1,138,69,70),(1,276,138,139).$$
Hence in total, taking into account overlapping in the computation of $h_{m_0P}$ and $h_{(m_0+m_1)P}$, Algorithm 2 calls Subalgorithm 1 only twelve times. 
\end{example}


\appendix
\section{Explicit formulas}
{$\textbf{(1)}$ Formulas for the coefficients of $S_{P,Q}$ in terms of the coefficients of $h_P$ and $h_Q$.\\
\newline
$a_4 = -\alpha_1^3\beta_1\beta_0^2 - 3B\alpha_1^3\beta_1 + 2A\alpha_1^3\beta_0 + 2\alpha_1^2\alpha_0\beta_1^2\beta_0 + A\alpha_1^2\alpha_0\beta_1 - 6B\alpha_1^2\beta_1^2 + 3A\alpha_1^2\beta_1\beta_0 + A^2\alpha_1^2 - \alpha_1\alpha_0^2\beta_1^3 + 6\alpha_1\alpha_0^2\beta_0 + 3A\alpha_1\alpha_0\beta_1^2 + 3\alpha_1\alpha_0\beta_0^2 + 9B\alpha_1\alpha_0 - 3B\alpha_1\beta_1^3 + A\alpha_1\beta_1^2\beta_0 + 2A^2\alpha_1\beta_1 - 3\alpha_1\beta_0^3 + 9B\alpha_1\beta_0 - 3\alpha_0^3\beta_1 + 3\alpha_0^2\beta_1\beta_0 - 3A\alpha_0^2 + 2A\alpha_0\beta_1^3 + 6\alpha_0\beta_1\beta_0^2 + 9B\alpha_0\beta_1 - 6A\alpha_0\beta_0 + A^2\beta_1^2 + 9B\beta_1\beta_0 - 3A\beta_0^2$\\
\newline
$a_3 = 4B\alpha_1^3\beta_1^3 - 2A\alpha_1^3\beta_1^2\beta_0 + A^2\alpha_1^3\beta_1 - \alpha_1^3\beta_0^3 + 9B\alpha_1^3\beta_0 - 2A\alpha_1^2\alpha_0\beta_1^3 - \alpha_1^2\alpha_0\beta_1\beta_0^2 + 3B\alpha_1^2\alpha_0\beta_1 - 7A\alpha_1^2\alpha_0\beta_0 + A^2\alpha_1^2\beta_1^2 - 6B\alpha_1^2\beta_1\beta_0 + 
3A\alpha_1^2\beta_0^2 + 6AB\alpha_1^2 - \alpha_1\alpha_0^2\beta_1^2\beta_0 + A\alpha_1\alpha_0^2\beta_1 - 6B\alpha_1\alpha_0\beta_1^2 + 12A\alpha_1\alpha_0\beta_1\beta_0 -
8A^2\alpha_1\alpha_0 + A^2\alpha_1\beta_1^3 + 3B\alpha_1\beta_1^2\beta_0 + A\alpha_1\beta_1\beta_0^2 - 6AB\alpha_1\beta_1 + 4A^2\alpha_1\beta_0 - 
\alpha_0^3\beta_1^3 - 3\alpha_0^3\beta_0 + 3A\alpha_0^2\beta_1^2 + 21\alpha_0^2\beta_0^2 - 18B\alpha_0^2 + 9B\alpha_0\beta_1^3 - 
7A\alpha_0\beta_1^2\beta_0 + 4A^2\alpha_0\beta_1 - 3\alpha_0\beta_0^3 + 18B\alpha_0\beta_0 + 6AB\beta_1^2 - 8A^2\beta_1\beta_0 - 18B\beta_0^2 
+ 4A^3 + 27B^2$\\
\newline
$a_2 = -A^2\alpha_1^3\beta_1^3 - 2A\alpha_1^3\beta_1\beta_0^2 - 6AB\alpha_1^3\beta_1 + A^2\alpha_1^3\beta_0 - 2A\alpha_1^2\alpha_0\beta_1^2\beta_0 + 5A^2\alpha_1^2\alpha_0\beta_1 - 3\alpha_1^2\alpha_0\beta_0^3 - 9B\alpha_1^2\alpha_0\beta_0 + 6AB\alpha_1^2\beta_1^2 + 9B\alpha_1^2\beta_0^2 + (2A^3+ 27B^2)\alpha_1^2 - 2A\alpha_1\alpha_0^2\beta_1^3 - 3\alpha_1\alpha_0^2\beta_1\beta_0^2 + 9B\alpha_1\alpha_0^2\beta_1 + 9A\alpha_1\alpha_0^2\beta_0 + 36B\alpha_1\alpha_0\beta_1\beta_0 - 12A\alpha_1\alpha_0\beta_0^2 - 18AB\alpha_1\alpha_0 - 6AB\alpha_1\beta_1^3 + 5A^2\alpha_1\beta_1^2\beta_0 + 9B\alpha_1\beta_1\beta_0^2 + (4A^3 - 27B^2)\alpha_1\beta_1 - 3A\alpha_1\beta_0^3 + 36AB\alpha_1\beta_0 - 3\alpha_0^3\beta_1^2\beta_0 - 3A\alpha_0^3\beta_1 + 9B\alpha_0^2\beta_1^2 - 12A\alpha_0^2\beta_1\beta_0 + 6A^2\alpha_0^2 + A^2\alpha_0\beta_1^3 - 9B\alpha_0\beta_1^2\beta_0 + 9A\alpha_0\beta_1\beta_0^2 + 36AB\alpha_0\beta_1 - 24A^2\alpha_0\beta_0 + (2A^3 + 27B^2)\beta_1^2 - 18AB\beta_1\beta_0 + 6A^2\beta_0^2$\\
\newline
$a_1 = -A^2\alpha_1^3\beta_1^2\beta_0 - 4B\alpha_1^3\beta_1\beta_0^2 + (A^3 - 12B^2)\alpha_1^3\beta_1 + 8AB\alpha_1^3\beta_0 - 
A^2\alpha_1^2\alpha_0\beta_1^3 - 4B\alpha_1^2\alpha_0\beta_1^2\beta_0 + 16AB\alpha_1^2\alpha_0\beta_1 - 3A^2\alpha_1^2\alpha_0\beta_0 + (-3A^3 + 
12B^2)\alpha_1^2\beta_1^2 - 24AB\alpha_1^2\beta_1\beta_0 + 3A^2\alpha_1^2\beta_0^2 - 2A^2B\alpha_1^2 - 4B\alpha_1\alpha_0^2\beta_1^3 -
3A^2\alpha_1\alpha_0^2\beta_1 - 3\alpha_1\alpha_0^2\beta_0^3 + 15B\alpha_1\alpha_0^2\beta_0 - 24AB\alpha_1\alpha_0\beta_1^2 + 12A^2\alpha_1\alpha_0\beta_1\beta_0 - 
6B\alpha_1\alpha_0\beta_0^2 - 18B^2\alpha_1\alpha_0 + (A^3 - 12B^2)\alpha_1\beta_1^3 + 16AB\alpha_1\beta_1^2\beta_0 - 3A^2\alpha_1\beta_1\beta_0^2
+ 14A^2B\alpha_1\beta_1 - 3B\alpha_1\beta_0^3 + 63B^2\alpha_1\beta_0 - 3\alpha_0^3\beta_1\beta_0^2 - 3B\alpha_0^3\beta_1 - 3A\alpha_0^3\beta_0 + 
3A^2\alpha_0^2\beta_1^2 - 6B\alpha_0^2\beta_1\beta_0 + 9A\alpha_0^2\beta_0^2 + 6AB\alpha_0^2 + 8AB\alpha_0\beta_1^3 - 
3A^2\alpha_0\beta_1^2\beta_0 + 15B\alpha_0\beta_1\beta_0^2 + 63B^2\alpha_0\beta_1 - 3A\alpha_0\beta_0^3 - 42AB\alpha_0\beta_0 - 2A^2B\beta_1^2
- 18B^2\beta_1\beta_0 + 6AB\beta_0^2 + 4A^4 + 27AB^2$\\
\newline
$a_0 = 2A^2B\alpha_1^3\beta_1 - A^3\alpha_1^3\beta_0 - A^2\alpha_1^2\alpha_0\beta_1^2\beta_0 - 4B\alpha_1^2\alpha_0\beta_1\beta_0^2 + 12B^2\alpha_1^2\alpha_0\beta_1 -4AB\alpha_1^2\alpha_0\beta_0 - 5A^2B\alpha_1^2\beta_1^2 + (A^3 - 24B^2)\alpha_1^2\beta_1\beta_0 + 6AB\alpha_1^2\beta_0^2 + (A^4 + 6AB^2)\alpha_1^2 - 4B\alpha_1\alpha_0^2\beta_1^2\beta_0 + 2A\alpha_1\alpha_0^2\beta_1\beta_0^2 - 2AB\alpha_1\alpha_0^2\beta_1 - A^2\alpha_1\alpha_0^2\beta_0+ (A^3 - 24B^2)\alpha_1\alpha_0\beta_1^2 + 24AB\alpha_1\alpha_0\beta_1\beta_0 - 3A^2\alpha_1\alpha_0\beta_0^2 + A^2B\alpha_1\alpha_0 + 2A^2B\alpha_1\beta_1^3 + 12B^2\alpha_1\beta_1^2\beta_0 - 2AB\alpha_1\beta_1\beta_0^2 + (-2A^4 - 6AB^2)\alpha_1\beta_1 - 
5A^2B\alpha_1\beta_0 - \alpha_0^3\beta_0^3 - 3B\alpha_0^3\beta_0 + 6AB\alpha_0^2\beta_1^2 - 3A^2\alpha_0^2\beta_1\beta_0 + 3B\alpha_0^2\beta_0^2 
+ (A^3 + 9B^2)\alpha_0^2 - A^3\alpha_0\beta_1^3 - 4AB\alpha_0\beta_1^2\beta_0 - A^2\alpha_0\beta_1\beta_0^2 - 5A^2B\alpha_0\beta_1 - 
3B\alpha_0\beta_0^3 + (2A^3 - 9B^2)\alpha_0\beta_0 + (A^4 + 6AB^2)\beta_1^2 + A^2B\beta_1\beta_0 + (A^3 + 
9B^2)\beta_0^2 + 4A^3B + 27B^3$\\
\newline
$b_3 = \alpha_1^3\beta_0^2 - B\alpha_1^3 - 2\alpha_1^2\alpha_0\beta_1\beta_0 + A\alpha_1^2\alpha_0 + \alpha_1^2\beta_1\beta_0^2 - 3B\alpha_1^2\beta_1 + A\alpha_1^2\beta_0 + \alpha_1\alpha_0^2\beta_1^2 - 2\alpha_1\alpha_0\beta_1^2\beta_0 + 2A\alpha_1\alpha_0\beta_1 - 3B\alpha_1\beta_1^2 + 2A\alpha_1\beta_1\beta_0 + \alpha_0^3 + \alpha_0^2\beta_1^3 + 3\alpha_0^2\beta_0 + A\alpha_0\beta_1^2 + 3\alpha_0\beta_0^2 - B\beta_1^3 + A\beta_1^2\beta_0 + \beta_0^3$\\
\newline
$b_2 = A^2\alpha_1^3 + 3\alpha_1^2\alpha_0\beta_0^2 + 9B\alpha_1^2\alpha_0 + 3A^2\alpha_1^2\beta_1 + 3\alpha_1^2\beta_0^3 + 9B\alpha_1^2\beta_0 - 6\alpha_1\alpha_0^2\beta_1\beta_0 - 3A\alpha_1\alpha_0^2 - 6\alpha_1\alpha_0\beta_1\beta_0^2 + 18B\alpha_1\alpha_0\beta_1 - 6A\alpha_1\alpha_0\beta_0 + 3A^2\alpha_1\beta_1^2 + 18B\alpha_1\beta_1\beta_0 - 3A\alpha_1\beta_0^2 + 3\alpha_0^3\beta_1^2 + 3\alpha_0^2\beta_1^2\beta_0 - 3A\alpha_0^2\beta_1 + 9B\alpha_0\beta_1^2 - 6A\alpha_0\beta_1\beta_0 + A^2\beta_1^3 + 9B\beta_1^2\beta_0 - 3A\beta_1\beta_0^2$\\
\newline
$b_1 = -A^2\alpha_1^3\beta_1^2 - 2A\alpha_1^3\beta_0^2 + 2AB\alpha_1^3 - 12B\alpha_1^2\alpha_0\beta_1^2 + 4A\alpha_1^2\alpha_0\beta_1\beta_0 - 3A^2\alpha_1^2\alpha_0 - A^2\alpha_1^2\beta_1^3 - 12B\alpha_1^2\beta_1^2\beta_0 + 4A\alpha_1^2\beta_1\beta_0^2 + A^2\alpha_1^2\beta_0 + 4A\alpha_1\alpha_0^2\beta_1^2 - 3\alpha_1\alpha_0^2\beta_0^2 - 9B\alpha_1\alpha_0^2 + 4A\alpha_1\alpha_0\beta_1^2\beta_0 + 2A^2\alpha_1\alpha_0\beta_1 + 6\alpha_1\alpha_0\beta_0^3 + 18B\alpha_1\alpha_0\beta_0 + 2A^2\alpha_1\beta_1\beta_0 + 9B\alpha_1\beta_0^2 + (4A^3 + 27B^2)\alpha_1 + 6\alpha_0^3\beta_1\beta_0 + A\alpha_0^3 - 2A\alpha_0^2\beta_1^3 - 3\alpha_0^2\beta_1\beta_0^2 + 9B\alpha_0^2\beta_1 - 9A\alpha_0^2\beta_0 + A^2\alpha_0\beta_1^2 + 18B\alpha_0\beta_1\beta_0 - 9A\alpha_0\beta_0^2 + 2AB\beta_1^3 - 3A^2\beta_1^2\beta_0 - 9B\beta_1\beta_0^2 + (4A^3 + 27B^2)\beta_1 + A\beta_0^3$\\
\newline
$b_0 = -2A^2\alpha_1^3\beta_1\beta_0 - 8B\alpha_1^3\beta_0^2 + (A^3 + 8B^2)\alpha_1^3 + A^2\alpha_1^2\alpha_0\beta_1^2 - 8B\alpha_1^2\alpha_0\beta_1\beta_0 + 6A\alpha_1^2\alpha_0\beta_0^2 - 2AB\alpha_1^2\alpha_0 + A^2\alpha_1^2\beta_1^2\beta_0 + 4B\alpha_1^2\beta_1\beta_0^2 + (-A^3 - 12B^2)\alpha_1^2\beta_1 + 4AB\alpha_1^2\beta_0 + 4B\alpha_1\alpha_0^2\beta_1^2 + A^2\alpha_1\alpha_0^2 - 2A^2\alpha_1\alpha_0\beta_1^3 - 8B\alpha_1\alpha_0\beta_1^2\beta_0 + 8AB\alpha_1\alpha_0\beta_1 - 6A^2\alpha_1\alpha_0\beta_0 + (-A^3 - 12B^2)\alpha_1\beta_1^2 + 8AB\alpha_1\beta_1\beta_0 - 3A^2\alpha_1\beta_0^2 + 3\alpha_0^3\beta_0^2 + B\alpha_0^3 - 8B\alpha_0^2\beta_1^3 + 6A\alpha_0^2\beta_1^2\beta_0 - 3A^2\alpha_0^2\beta_1 + 3\alpha_0^2\beta_0^3 - 15B\alpha_0^2\beta_0 + 4AB\alpha_0\beta_1^2 - 6A^2\alpha_0\beta_1\beta_0 - 15B\alpha_0\beta_0^2 + (4A^3 + 27B^2)\alpha_0 + (A^3 + 8B^2)\beta_1^3 - 2AB\beta_1^2\beta_0 + A^2\beta_1\beta_0^2 + B\beta_0^3 + (4A^3 + 27B^2)\beta_0$\\
\newline
{$\textbf{(2)}$ Doubling formulas for $h_P$. Write $h_{2P}=cy-(u_0+u_1x)$, then: \\
\newline
$u_1 = 4B\alpha_1^4 - 4A\alpha_1^3\alpha_0 + 4A^2\alpha_1^2 - 4\alpha_1\alpha_0^3 + 36B\alpha_1\alpha_0 - 12A\alpha_0^2$\\
\newline
$u_0 = -A^2\alpha_1^4 - 8B\alpha_1^3\alpha_0 + 2A\alpha_1^2\alpha_0^2 + 6AB\alpha_1^2 - 8A^2\alpha_1\alpha_0 - \alpha_0^4 - 18B\alpha_0^2 + 4A^3 + 27B^2$\\
\newline
$c = 8B\alpha_1^3 - 8A\alpha_1^2\alpha_0 - 8\alpha_0^3$\\
\newline
{$\textbf{(3)}$ Tripling formulas for $h_P$. Write $h_{3P}=dy-(v_0+v_1x)$, then:\\
\newline
$v_1 = 1/3A^4\alpha_1^9 + 8A^2B\alpha_1^8\alpha_0 + (-4A^3 + 48B^2)\alpha_1^7\alpha_0^2 + (16A^3B + 
144B^3)\alpha_1^7 - 48AB\alpha_1^6\alpha_0^3 + (-16A^4 - 240AB^2)\alpha_1^6\alpha_0 + 10A^2\alpha_1^5\alpha_0^4 + 
192A^2B\alpha_1^5\alpha_0^2 + (8A^5 + 54A^2B^2)\alpha_1^5 - 24B\alpha_1^4\alpha_0^5 + (-112A^3 + 
144B^2)\alpha_1^4\alpha_0^3 + (96A^3B + 648B^3)\alpha_1^4\alpha_0 + 12A\alpha_1^3\alpha_0^6 - 240AB\alpha_1^3\alpha_0^4 +
(-48A^4 - 324AB^2)\alpha_1^3\alpha_0^2 + (-32A^4B - 216AB^3)\alpha_1^3 - 48A^2\alpha_1^2\alpha_0^5 + 
(64A^5 + 432A^2B^2)\alpha_1^2\alpha_0 + 3\alpha_1\alpha_0^8 - 288B\alpha_1\alpha_0^6 + (-24A^3 - 162B^2)\alpha_1\alpha_0^4
+ (288A^3B + 1944B^3)\alpha_1\alpha_0^2 + (-16A^6 - 216A^3B^2 - 729B^4)\alpha_1 + 48A\alpha_0^7 
+ (-64A^4 - 432AB^2)\alpha_0^3$\\
\newline
$v_0 = (-8/3A^3B - 64/3B^3)\alpha_1^9 + (3A^4 + 32AB^2)\alpha_1^8\alpha_0 - 16A^2B\alpha_1^7\alpha_0^2 - 8A^2B^2\alpha_1^7 + (12A^3 + 16B^2)\alpha_1^6\alpha_0^3 + (8A^3B - 144B^3)\alpha_1^6\alpha_0 + 
8AB\alpha_1^5\alpha_0^4 + 288AB^2\alpha_1^5\alpha_0^2 + (32A^4B + 216AB^3)\alpha_1^5 + 10A^2\alpha_1^4\alpha_0^5 
- 200A^2B\alpha_1^4\alpha_0^3 + (-24A^5 - 162A^2B^2)\alpha_1^4\alpha_0 + 32B\alpha_1^3\alpha_0^6 + (64A^3 + 
72B^2)\alpha_1^3\alpha_0^4 + (192A^3B + 1296B^3)\alpha_1^3\alpha_0^2 + (96A^3B^2 + 648B^4)\alpha_1^3 - 
4A\alpha_1^2\alpha_0^7 - 72AB\alpha_1^2\alpha_0^5 + (-176A^4 - 1188AB^2)\alpha_1^2\alpha_0^3 + (-192A^4B - 
1296AB^3)\alpha_1^2\alpha_0 + 64A^2\alpha_1\alpha_0^6 + (128A^5 + 864A^2B^2)\alpha_1\alpha_0^2 + 1/3\alpha_0^9 + 
72B\alpha_0^7 + (-120A^3 - 810B^2)\alpha_0^5 + (192A^3B + 1296B^3)\alpha_0^3 + (-16A^6 - 
216A^3B^2 - 729B^4)\alpha_0$\\
\newline
$d = A^4\alpha_1^8 + 24A^2B\alpha_1^7\alpha_0 + (-12A^3 + 144B^2)\alpha_1^6\alpha_0^2 + (-24A^3B - 
144B^3)\alpha_1^6 - 144AB\alpha_1^5\alpha_0^3 + (32A^4 + 144AB^2)\alpha_1^5\alpha_0 + 30A^2\alpha_1^4\alpha_0^4 + 
120A^2B\alpha_1^4\alpha_0^2 + (-8A^5 - 54A^2B^2)\alpha_1^4 - 72B\alpha_1^3\alpha_0^5 + 720B^2\alpha_1^3\alpha_0^3 +
(-96A^3B - 648B^3)\alpha_1^3\alpha_0 + 36A\alpha_1^2\alpha_0^6 - 360AB\alpha_1^2\alpha_0^4 + (48A^4 + 
324AB^2)\alpha_1^2\alpha_0^2 + 96A^2\alpha_1\alpha_0^5 + 9\alpha_0^8 + 72B\alpha_0^6 + (24A^3 + 162B^2)\alpha_0^4 - 
16/3A^6 - 72A^3B^2 - 243B^4$


\begin{thebibliography}{91}

\bibitem{handbook} R. M. Avanzi, H. Cohen, C. Doche, G. Frey, T. Lange, K. Nguyen, F. Vercauteren, {\em Handbook of Elliptic and Hyperelliptic Curve Cryptography}, Discrete Mathematics and Its Applications 34, Chapman \& Hall/CRC (2005).

\bibitem{ac07} R. M. Avanzi, E. Cesena, {\em Trace  zero varieties over  fields  of characteristic  2 for  cryptographic
applications}, Proceedings of the First Symposium on Algebraic Geometry and Its Applications -- SAGA '07 (2007), 188-215.


\bibitem{TrZeroAvanzi} R. M.  Avanzi, E. Cesena, T. Lange, {\em Trace Zero Varieties for Cryptographic Applications}, SPEED-CC, Berlin, October 13th, 2009.









\bibitem{ces08} E. Cesena, {\em Pairing with Supersingular Trace Zero Varieties Revisited}, Available at http: // eprint.iacr.org/2008/404,2008.






\bibitem{DS2} C. Diem, J. Scholten, {\em An attack on a trace-zero cryptosystem}, Available at \url{http://www.math.uni-leipzig.de/diem/preprints}.

\bibitem{Frey} G. Frey,
\emph{Applications of Arithmetical Geometry to Cryptographic Constructions}, Proceedings of the 5th International Conference on Finite Fields and Applications, Springer (1999),128-161.



\bibitem{kob} N. Koblitz, \emph{CM-curves with good cryptographic properties}, Advances in Cryptology - Crypto 1991, Lecture Notes in Comput. Sci., vol. 576, Springer-Verlag, Berlin, 1992, 279-287.

\bibitem{EM1} E. Gorla, M. Massierer, 
\emph{Point Compression for the Trace Zero Subgroup over a Small Degree Extension Field}, 
Designs, Codes and Cryptography 75, no. 2 (2015), 335-357.

\bibitem{EM2} E. Gorla, M. Massierer,
\emph{An Optimal Representation for the Trace Zero Subgroup}, available at \url{http://arxiv.org/abs/1405.2733}.

\bibitem{langetzero} T. Lange, 
\emph{Trace zero subvarieties of genus $2$ curves for cryptosystem}, Ramanujan Math. Soc. 19, no. 1 (2004) 15-33.


\bibitem{Mont} P. L. Montgomery, 
\emph{Speeding the Pollard and elliptic curve methods of factorization}, Mathematics of Computation, 48(177):243:264, January 1987.

\bibitem{nau99} N. Naumann, {\em Weil-Restriktion abelscher Variet\"aten}, 
Master's thesis (1999), available at \url{http://web.iem.uni-due.de/ag/numbertheory/dissertationen}.

\bibitem{rs02} K. Rubin, A. Silverberg, {\em Supersingular abelian varieties in cryptology}, Advances in Cryptology: Proocedings of CRYPTO '02 (M. Young, ed), LNCS, vol. 2442, Springer, 2002, pp. 336-353. 


\bibitem{rs09} K. Rubin, A. Silverberg, {\em Using abelian varieties to improve pairing-based cryptography},
Journal of Cryptology 22, no. 3 (2009), 330-364.


\bibitem{sil05} A. Silverberg, {\em Compression for Trace Zero Subgroups of Elliptic Curves}, 
Trends in Mathematics 8 (2005), 93-100. 





\bibitem{Weim} A. Weimerskirch, {\em The application of the Mordell-Weil group to cryptographic systems}, Master's thesis, Worcester Polytechnic Institute, Available at 
\url{http :// www.emsec.rub.de/media/crypto/attachments/files/2010/04/ms-weika.pdf}, 2001. 


\end{thebibliography}
\end{document}